\documentclass[10 pt, journal]{IEEEtran}

\makeatletter
\def\ps@headings{%
	\def\@oddhead{\mbox{}\scriptsize\rightmark \hfil \thepage}%
	
	\def\@evenhead{\scriptsize\thepage \hfil \leftmark\mbox{}}%
	
	\def\@oddfoot{}%
	
	\def\@evenfoot{}}
	
\makeatother
\usepackage[hidelinks]{hyperref}

\usepackage[table]{xcolor}

\usepackage{float}  
\usepackage{epsfig,bm,amsmath,amssymb,graphicx,setspace}
\usepackage[numbers,sort&compress]{natbib}
\usepackage{color,placeins}
\usepackage{multirow} 
\usepackage{threeparttable}
\usepackage{dsfont}
\usepackage[linesnumbered,ruled,vlined]{algorithm2e}
\usepackage{algpseudocode}
\usepackage{colortbl}
\usepackage{color}
\usepackage{dblfloatfix}
\usepackage{turnstile}
\usepackage{multicol}
\usepackage{array}

\usepackage{enumerate}
\usepackage{turnstile}
\usepackage{subcaption}
\usepackage{arydshln,paralist}
\usepackage{soul,tabularx,booktabs}
\usepackage{csquotes}
\usepackage{multirow}
\usepackage{adjustbox}
\usepackage{float}

\usepackage[n,advantage,operators,sets,adversary,landau,probability,notions,logic,ff,mm,primitives,events,complexity,asymptotics,keys]{cryptocode}
\usepackage{amsmath}
\usepackage{amssymb}

\usepackage{xcolor}
\usepackage{verbatim}
\usepackage[colorinlistoftodos]{todonotes} %Add disable parameter to turn off comments
\usepackage{rotating}
\definecolor{usethiscolorhere}{rgb}{0.86666,0.78431,0.78431}
\usepackage{tikz}
\usetikzlibrary{mindmap}
\usetikzlibrary{calc,positioning}
\usepackage{lipsum,adjustbox}
\usetikzlibrary{decorations.pathmorphing}

\usepackage{pifont}% http://ctan.org/pkg/pifont
\newcommand{\cmark}{\ding{51}}%
\newcommand{\xmark}{\ding{55}}%

\usepackage{amsthm}
\newtheorem{theorem}{Theorem}
\newtheorem{property}{Property}

\newcommand{\specialcell}[2][c]{% 
	\begin{tabular}[#1]{@{}c@{}}#2\end{tabular}}

\makeatother

\usepackage{amsmath}

\begin{document}

\title{\fontsize{24pt}{26pt}\selectfont \textit{LiteAtt}: A Peer-to-Peer Self-Attestation Framework and Handshake Protocol for Connected IoT Devices}

% \author{
%     \IEEEauthorblockN{Varun Kohli$^\dagger$, Muhammad Naveed Aman$^\ddagger$,~\IEEEmembership{Senior Member,~IEEE}, Biplab Sikdar$^\dagger$,~\IEEEmembership{Fellow,~IEEE}
%     \IEEEauthorblockA{$^\dagger$Department of Electrical and Computer Engineering, National University of Singapore, Singapore.}} \\
%     \IEEEauthorblockA{$^\ddagger$School of Computing, University of Nebraska-Lincoln, USA.}
% }

\author{Varun Kohli,~\IEEEmembership{Graduate Student Member,~IEEE}, Biplab Sikdar,~\IEEEmembership{Fellow,~IEEE}%
\thanks{V. Kohli is with the Institute for Infocomm Research ($I^2R$), Agency for Science, Technology and Research (A*STAR), 1 Fusionopolis Way, \#21-01 Connexis, Singapore 138632, and the Department of Electrical and Computer Engineering, National University of Singapore, Singapore 117417 (email: kohliv@a-star.edu.sg, varun.kohli@u.nus.edu).}
\thanks{B. Sikdar is with the Department of Electrical and Computer Engineering, National University of Singapore, Singapore 117417 (e-mail: bsikdar@nus.edu.sg).}}

\maketitle

\begin{abstract}
As the Internet of Things (IoT) becomes an integral part of critical infrastructure and commercial services, runtime firmware attestation of constituent Micro-Controllers (MCUs) has become instrumental in maintaining security and trust. Most prior works assume computational limitations on the MCUs and rely on a remote verifier to perform complex computation. This introduces a centralized point of failure, round-trip latency, and the burden of maintaining golden reference states at the recipient, even in recent Peer-to-Peer (P2P) and Self-Attestation (SA) schemes. This is avoidable for modern MCUs such as Arm Cortex-M, which, although battery-operated, feature security and intelligence capabilities, including Trusted Execution Environments (TEE) and embedded Tiny Machine Learning (TinyML) inference. Leveraging such provisions, this paper presents \textit{LiteAtt}, a verifier-less, P2P-SA framework and protocol for modern IoT MCUs that folds directly into the connection handshake between IoT MCUs. Each MCU runs a quantized TinyML Autoencoder (TinyAE) within its TEE to evaluate the runtime SRAM state. SA verdicts are securely bound to the handshake transcript context, enabling stateless verification at the peer node. The proposed protocol yields mutually authenticated and firmware-attested communication without traditional latency and reference distribution overheads while ensuring mutual authentication, forward secrecy, confidentiality, integrity, SRAM privacy, and defense against replay, SA report spoofing, and time-of-check-time-of-use (TOCTOU) impersonation attacks. We report an optimized per-handshake latency, energy consumption, and peak memory footprint of 26.3-294.9ms, 2.65-9.35mJ, and 4.91KB, respectively, across three Arm Cortex-M boards. Further, the suggested TinyAE model achieves an average accuracy of 99.42\%, F1 score of 99.70\%, TPR of 99.45\%, and TNR of 95.14\% on SRAM attestation datasets covering single-node and swarm sensing, actuation, and cryptographic IoT applications. We also evaluate the models on simulated runtime attacks, random perturbation, and adversarial ML.
\end{abstract}


\section{Introduction}
\label{sec:intro}

The Internet of Things (IoT) has paved its way into critical infrastructure and the consumer landscape, with millions of microcontrollers (MCUs) deployed in homes, industry, healthcare, transportation, energy, and security networks across the world to facilitate data collection and processing, automate critical operations, and enhance our quality of life \cite{sethi2017internet}. This has led to threat actors leveraging inadequate security provisions in IoT devices via software attacks to disrupt operations or extort data for monetary gains \cite{wetzels2023insecure,hassija2019survey}, such as the Mirai Botnet incident \cite{antonakakis2017understanding}. To address such threats to security, user privacy, and critical infrastructure operations, researchers have proposed several approaches to verify the integrity of IoT device firmware via Remote Attestation (RA).

As depicted in Fig. \ref{fig:typical_network}, RA typically involves a request-response routine wherein a remote \textit{verifier}, i.e., a trusted device operated by an external security operator, sends an attestation request to a \textit{prover}, i.e., a vulnerable, user-operated device that in turn measures the state of the system and sends it to the verifier for evaluation. Existing RA methods differ primarily in the evidence the prover collects. Program-memory hashing schemes compute one or more hash iterations over the MCU's flash memory \cite{seshadri2004swatt,seshadri2005pioneer}, optionally combined with timing measurements \cite{castelluccia2009difficulty} or hardware roots of trust \cite{tan2011tpm,agrawal2015program,aman2020hatt}. Control-Flow Attestation (CFA) instead attests the runtime execution path using complete or partial Control-Flow Graph (CFG), and Control-Flow Path (CFP) hashes \cite{yadav2023whole,chilese2024one,ammar2025sok}, with recent works anchoring evidence delivery in a Trusted Execution Environment (TEE) \cite{caulfield2024traces}. Finally, Static Random Access Memory (SRAM)-based methods \cite{aman2022machine,iqbal2024ram,kohli2024swarm} attest firmware during runtime using their SRAM footprints.

\begin{figure}[t]
    \centering
    \includegraphics[width=0.7\linewidth]{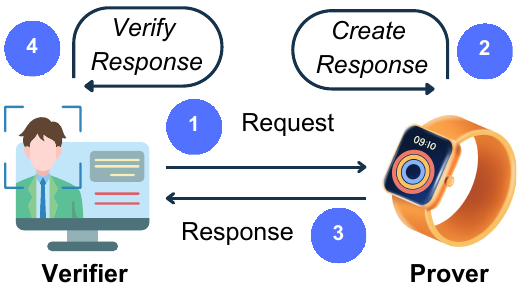}
    \caption{Typical network model for remote firmware attestation. A central \textit{verifier} remotely attests each \textit{prover} in the network. \textit{LiteAtt} moves away from this toward decentralized, P2P-SA.}
    \label{fig:typical_network}
\end{figure}

The reliance on an external verifier is, however, a structural limitation of the RA model, introducing four practical problems: \textit{\textbf{First}}, as a trusted central entity, the verifier becomes a centralized point of failure that must remain secure and accessible to remotely deployed provers. \textit{\textbf{Second}}, RA adds round-trip latency to every attestation coordinated with the central verifier before any Peer-to-Peer (P2P) communication can proceed, since the prover must collect and transmit evidence and the verifier must perform comparison with up-to-date reference states, hashing, or ML inference. \textit{\textbf{Third}}, sharing runtime evidence directly exposes potentially sensitive data \cite{aman2022machine,iqbal2024ram,kohli2024safe}. These limitations make it difficult to integrate RA directly into IoT services, and although existing work on P2P \cite{ferro2024safehive,shepherd2021lira} and Self-Attestation (SA) \cite{zhang2025flashattest,ibrahim2017seed} solve some of these issues, they typically require the recipient to maintain copies of the latest golden reference states of other devices, which is not scalable in large and dynamic IoT deployments.

Efforts to create powerful, lightweight MCUs have made P2P-SA with stateless recipients a practically feasible attestation model thanks to larger memory, quantized embedded Tiny Machine Learning (TinyML) for edge inference \cite{lin2023tiny}, and Trusted Execution Environments (TEE) for isolated and secure execution of critical software components \cite{jauernig2020trusted, caulfield2024traces}. Furthermore, IoT devices are typically battery-operated and make tens to thousands of connections per day. Among the runtime evidences used in prior work, we find SRAM-based evidence to be particularly well-suited to P2P-SA for the following reasons: \raisebox{.5pt}{\textcircled{\raisebox{-.8pt} {1}}} Smaller size, faster traversal and processing compared to flash memory and control-flow evidence. \raisebox{.5pt}{\textcircled{\raisebox{-.8pt} {2}}} Captures runtime information, exposing code modification and runtime threats. \raisebox{.5pt}{\textcircled{\raisebox{-.8pt} {3}}} Available on every MCU and does not require additional hardware-assisted logging. \raisebox{.5pt}{\textcircled{\raisebox{-.8pt} {4}}} Deterministic footprint of the same firmware on physical and digital twin hardware \cite{kohli2024intelligent, kohli2024swarm}. \raisebox{.5pt}{\textcircled{\raisebox{-.8pt} {5}}} Utility in attestation has been demonstrated by prior work \cite{iqbal2024ram, kohli2024swarm}.
 
\textbf{\textit{Contributions.}} This paper addresses the limitations of prior work on RA, P2P, and SA by proposing a novel P2P-SA approach that relies on TinyML, TEE, and SRAM, folding directly into the communication handshake between IoT devices. It makes the following contributions:

\begin{enumerate}
    \item A lightweight P2P-SA framework called \textit{LiteAtt} is proposed for decentralized runtime attestation, folded seamlessly into the connection handshakes between devices in critical infrastructure and commercial services to ensure \textit{``firmware-attested"} IoT services. Devices run a secure application ($\text{App}_{SA}$) within the TEE, generating a signed SA report ($\sigma$) based on the evaluation of truncated, Discrete Cosine (DCT)-compressed SRAM data, heap, and stack evidence using TinyML models during runtime. \textit{LiteAtt} is lightweight, scalable, and leverages SRAM properties of digital/physical twin devices for easy post-deployment updates and privacy. 
    \item A secure P2P-SA protocol is proposed that ensures mutual authentication, forward secrecy, message security, SRAM privacy, and defense against replay and impersonation attacks. A game-based security analysis is presented to prove these properties, and sensitivity analysis is conducted on the TinyML models using random and adversarial ML strategies.
    \item Discriminative performance and TinyML robustness is demonstrated using SRAM data collected from sensor, processing, actuation, and cryptographic applications on real IoT boards, covering 23 safe firmware, 34 modifications to data and functional dependencies, and 7 data injection attacks. We also simulate Data-oriented Programming (DOP), Return-oriented Programming (ROP), Jump-oriented Programming (JOP) attacks, random perturbation and adversarial ML for sensitivity analysis.
    \item Thorough runtime experiments are conducted using real Arm Cortex-M MCU boards, including a TEE-enabled Arduino Portenta C33, a dual-processor Arduino Portenta H7, and an Arduino Nano 33 BLE Sense to highlight attestation latency, memory, and energy overheads across various MCU capabilities.
\end{enumerate}

The remainder of this paper is organized as follows: Section \ref{sec:related} highlights the research gap in existing works. Section \ref{sec:design} discusses our goals for \textit{LiteAtt}. Section \ref{sec:background} discusses the key foundational concepts of the proposed method. Section \ref{sec:network} introduces the network and threat model. Section \ref{sec:proposed} presents the proposed method, followed by Section \ref{sec:setup}, which details the experimental testbed. Sections \ref{sec:results} and \ref{sec:security} present our results and security analysis, respectively. Section \ref{sec:limitations} highlights the limitations, and the paper concludes in Section \ref{sec:conclusion}.

\begin{table*}[t]
\centering
\caption{Comparison of representative attestation approaches with \textit{LiteAtt} in terms of attestation model, evidence types, attestation method, runtime attack coverage, evidence overhead on prover, stateless recipients, ease of post-deployment updates, and P2P integration by design.}
\label{tab:related}
\small
\setlength{\tabcolsep}{5pt}
\renewcommand{\arraystretch}{1.15}
\resizebox{\linewidth}{!}{\begin{tabular}{l c c c c c c c c}
\toprule
\textbf{Reference} & \textbf{Model} & \textbf{Evidence} & \textbf{Method} & \textbf{Runtime} & \textbf{Evidence} & \textbf{Stateless} & \textbf{Easy} & \textbf{P2P} \\
& & & & \textbf{Attacks} & \textbf{Overhead} & \textbf{Recipient} & \textbf{Update} & \textbf{Design} \\
\midrule
SWATT \cite{seshadri2004swatt}             & RA   & Flash         & Hash, timing       & \xmark & High    & \xmark & \xmark & \xmark \\
Pioneer \cite{seshadri2005pioneer}         & RA   & Flash         & Hash, timing       & \xmark & High    & \xmark & \xmark & \xmark \\
SCUBA \cite{seshadri2006scuba}             & RA   & Flash         & Hash                & \xmark & High    & \xmark & \xmark & \xmark \\
SAKE \cite{seshadri2008sake}               & RA   & Flash         & Hash                & \xmark & High    & \xmark & \xmark & \xmark \\
TPM-RA \cite{tan2011tpm,agrawal2015program}& RA   & Flash         & Signed hash         & \xmark & High    & \xmark & \xmark & \xmark \\
TyTAN \cite{brasser2015tytan}              & RA   & Identity      & MAC                 & \xmark & Low     & \xmark & \xmark & \xmark \\
HAtt \cite{aman2020hatt}                   & RA   & Flash         & Hash, PUF          & \cmark & High    & \xmark & \xmark & \xmark \\
SMARM \cite{carpent2018remote}             & RA   & Flash         & Shuffled hash       & \cmark & High    & \xmark & \xmark & \xmark \\
C-FLAT \cite{abera2016c}                   & RA   & CFP           & CFP hash           & \cmark & High    & \xmark & \xmark & \xmark \\
BLAST \cite{yadav2023whole}                & RA   & CFG           & Whole-CFG hash      & \cmark & High    & \xmark & \xmark & \xmark \\
RAGE \cite{chilese2024one}                 & RA   & CFG           & VGAE inference      & \cmark & High    & \xmark & \xmark & \xmark \\
TRACES \cite{caulfield2024traces}          & RA   & CFP           & TEE-logged trace    & \cmark & High    & \xmark & \xmark & \xmark \\
Aman et al. \cite{aman2022machine}         & RA   & SRAM          & MLP inference       & \cmark & Low     & \xmark & \xmark & \xmark \\
Iqbal et al. \cite{iqbal2024ram}           & RA   & SRAM          & VAE inference       & \cmark & Low     & \xmark & \xmark & \xmark \\
Swarm-Net \cite{kohli2024swarm}            & RA   & SRAM          & GT inference       & \cmark & Low     & \xmark & \xmark & \xmark \\
SAFE-IoT \cite{kohli2024safe}              & RA   & SRAM          & MoE inference       & \cmark & Low     & \xmark & \xmark & \xmark \\
\midrule
SEED \cite{ibrahim2017seed}                & SA   & Flash         & Periodic hash, MAC & \xmark & High    & \xmark & \xmark & \xmark \\
ERASMUS \cite{carpent2018erasmus}          & SA   & Flash         & Periodic hash log   & \xmark & High    & \xmark & \xmark & \xmark \\
SARA \cite{dushku2020sara}                 & SA   & Flash         & Timestamped hash    & \xmark & High    & \xmark & \xmark & \xmark \\
SIMPLE \cite{ammar2020simple}              & SA   & Flash         & HMAC                & \xmark & High    & \xmark & \xmark & \xmark \\
FlashAttest \cite{zhang2025flashattest}    & SA   & Flash         & Timestamped hash    & \xmark & High    & \xmark & \xmark & \xmark \\
SAFEHIVE \cite{ferro2024safehive}          & SA   & Flash         & DHT-stored hash   & \xmark & High    & \xmark & \xmark & \cmark \\
LIRA-V \cite{shepherd2021lira}             & SA   & Flash         & Signed hash         & \xmark & High    & \xmark & \xmark & \cmark \\
TLS-Att \cite{tschofenig2025using}         & SA   & Agnostic & Signed quote  & \specialcell{Evidence\\dependent} & \specialcell{Evidence\\dependent}  & \xmark & \xmark & \cmark \\
\midrule
\textbf{\textit{LiteAtt} (ours)}           & \textbf{SA} & \textbf{SRAM} & \textbf{TinyML in TEE} & \cmark & \textbf{Low} & \cmark & \cmark & \cmark \\
\bottomrule
\end{tabular}}
\end{table*}

\section{Related Works}
\label{sec:related}

This section discusses the main evidence categories used in RA, followed by an introduction to existing SA and P2P approaches, and the limitations of each category concerning seamless P2P integration in IoT networks with dynamic topologies and frequent firmware updates. Table \ref{tab:related} summarizes the key comparison of prior works with concerning P2P-SA practicality. We also discuss design goals for \textit{LiteAtt}. 

\subsection{Flash-Memory Attestation}
Flash-memory attestation produces evidence by hashing the device's program memory. SWATT \cite{seshadri2004swatt} pioneered this approach by computing several checksum iterations using pseudo-random traversal, requiring precise timing verification by the verifier and incurring a prover runtime of up to 85 seconds and high energy cost. Castelluccia et al. \cite{castelluccia2009difficulty} highlighted the difficulty of timing-based software RA, noting the vulnerability to proxy attacks and the challenges of variable-latency networks. SCUBA \cite{seshadri2006scuba}, SAKE \cite{seshadri2008sake}, and Pioneer \cite{seshadri2005pioneer} share similar problems and are susceptible to Time-of-Check-to-Time-of-Use (TOCTOU) attacks, in which transient malware erases itself between successive attestations \cite{de2021toctou}. Trusted Platform Module (TPM)-anchored variants \cite{tan2011tpm, agrawal2015program} leverage tamper-resistant hardware to strengthen these guarantees, while hybrid schemes reduce hardware cost. For instance, TyTAN \cite{brasser2015tytan} provides a tiny trust anchor for isolated task loading and Message Authentication Code (MAC)-based identity attestation, and HAtt \cite{aman2020hatt} uses a Physically Unclonable Function (PUF) to anchor secrets and randomizes the selection of flash memory blocks to detect roving malware. SMARM \cite{carpent2018remote} shuffles memory measurements that catch roving malware, but its attestation routine takes over 50 seconds on the prover. 

\textit{Limitations:} Flash memory-based methods incur high overheads on the prover, making them unsuitable for P2P attestation of devices that offer real-time services in critical infrastructure and commercial services. They also require the recipient (verifier or peer) to host the up-to-date reference hash states, which is impractical in large-scale, dynamic deployments and under frequent firmware updates. Further, attacks such as ROP, DOP, JOP, and data injection go undetected since they do not tamper with the firmware itself \cite{de2021toctou, ammar2025sok}.

\subsection{Control-flow Attestation}
CFA addresses the runtime gap by collecting evidence about the program's execution path rather than its static code \cite{ammar2025sok}. C-FLAT \cite{abera2016c} hashes the sequence of branch destinations executed at runtime which the verifier compares against precomputed golden reference hashes derived from the program's CFG. BLAST \cite{yadav2023whole} extends this to whole-program CFGs. RAGE \cite{chilese2024one} replaces exact path matching with a Variational Graph Autoencoder (VGAE) over partial CFGs, achieving 91\% and 98\% detection for ROP and DOP, respectively. TRACES \cite{caulfield2024traces} leverages Arm TrustZone-M to guarantee delivery of periodic runtime reports even from compromised provers. 

\textit{Limitations:} Trace logging introduces nontrivial prover overhead that hinders real-time, seamless P2P attestation. Further, most CFA techniques require the recipient to maintain the latest reference states of ML models for the latest firmware state of the prover.

\subsection{Runtime SRAM Attestation}
SRAM-based attestation collects volatile memory states rather than static code or control-flow traces, capturing both code modifications and runtime attacks that perturb global, static, and dynamically allocated data, or function and interrupt return addresses. Aman et al. \cite{aman2022machine} pioneered this using ML classifiers trained on SRAM traces to distinguish between normal and malicious samples at 96\% accuracy. Iqbal et al. \cite{iqbal2024ram} improved upon this using Variational Autoencoders (VAE), achieving 100\% accuracy on their single-node dataset \cite{0nze-r023-24}. Swarm-Net \cite{kohli2024swarm} used Graph Transformers (GT) for swarm RA, achieving 99.7\% accuracy on node, network, and propagated anomalies on their swarm SRAM dataset \cite{gmee-vj41-24}. SAFE-IoT \cite{kohli2024safe} used a Mixture-of-Experts (MoE) architecture for swarm RA, achieving 95\% accuracy.

\textit{Limitations:} While lightweight, all prior work transmits raw SRAM data to the verifier, exposing potentially sensitive user data to third parties, and requiring ML models for inference on the recipient. Furthermore, some approaches rely only on the data section (.data and .bss) of the SRAM \cite{kohli2024safe,kohli2024swarm}, which miss runtime attacks that tamper with dynamic variables and function or interrupt return addresses.  

\subsection{P2P and Self-Attestation}
SA shifts the complex attestation computation from a recipient to the prover itself. Most existing SA approaches build on flash memory evidence. SEED \cite{ibrahim2017seed} periodically evaluates program hashes via a hardware circuit trigger. ERASMUS \cite{carpent2018erasmus}, based on the SMART architecture \cite{eldefrawy2012smart}, periodically attests flash memory at 0.5 s/KB. SARA \cite{dushku2020sara} uses a hardware-protected clock for timestamped flash hashes. SIMPLE \cite{ammar2020simple} attests flash memory at 0.26 s/KB using memory isolation. FlashAttest \cite{zhang2025flashattest} uses flash devices with authenticated timestamps at a 0.24 s/KB latency. A few works including SACHa \cite{vliegen2019sacha} and Usama et al. \cite{usama2024run} propose bitstream attestation and a Finite State Machine (FSM)-based approach, respectively, for FPGA runtime integrity.

Beyond these approaches, a parallel line of work has explored P2P attestation, where devices in a network validate each other's integrity without a dedicated verifier. SAFEHIVE \cite{ferro2024safehive} distributes reference values across IoT swarm members using a Distributed Hash Table (DHT). LIRA-V \cite{shepherd2021lira} proposes mutual attestation for constrained RISC-V devices, using program memory and the Physical Memory Protection (PMP) primitive as a trust anchor, achieving bi-directional attestation of 64-256KB of program memory in 11-32s. At the protocol level, \cite{tschofenig2025using} binds TEE-produced attestation evidence directly to the Transport Layer Security (TLS) 1.3 handshake, enabling peers to evaluate their counterpart's runtime state against a database of reference values during the connection handshake. 

\textit{Limitations:} Most prior work uses flash memory as evidence and inherits its limitations, and requires reference hash states on the recipient.

\section{Design Goals}
\label{sec:design}
Prior work has one or more limitations in terms of its evidence overhead, reference state maintenance on the recipient, attack coverage, data privacy, and applicability to P2P handshakes. We define the following design goals for \textit{LiteAtt}:
\begin{enumerate}
    \item \textbf{\textit{Decentralized, verifier-less P2P architecture:}} Eliminate the reliance on third-party verifiers in P2P settings.
    \item \textbf{\textit{Stateless recipient and scalability:}} Avoid requiring recipient devices to store, maintain, and check large databases of golden references or execute complex ML inference, which are fundamentally unscalable in large and dynamic IoT deployments with frequent firmware updates. Instead, trust shifts to the TEE and the certification provided by the IoT vendor's Certificate Authority (CA).
    \item \textbf{\textit{Comprehensive runtime threat coverage:}} Capture attacks that impact the runtime execution of the IoT firmware. These include code changes that change data and functional dependencies, ROP, DOP, and JOP attacks, and data injection.
    \item \textbf{\textit{Minimal resource and latency overheads:}} Provide lightweight processing suitable for real-time applications on battery-operated devices, bypassing the heavy overhead of continuous trace logging or flash traversal.
    \item \textbf{\textit{Data and SRAM privacy preservation:}} Avoid transmitting sensitive data or SRAM contents with the peer during attestation to protect user privacy.
    \item \textbf{\textit{Resilient post-deployment lifecycle updates:}} Ensure that post-deployment updates do not require global re-distribution of states across the rest of the network nodes.
\end{enumerate}

\textit{LiteAtt} addresses these problems by leveraging SRAM, TinyML and TEE: \textbf{First, at the protocol-level,} the attestation is folded into the connection handshake between two IoT devices. The SA outcome is computed inside the prover's TEE and encoded as a signed report, reducing the recipient's task to just validation without the need for it to maintain up-to-date ML model or golden reference states. \textbf{Second, at the evidence level,} the prover-side SRAM evidence is cheap to fetch and process. Data, heap and stack sections are evaluated using TEE-hosted TinyML models to capture a wide range of attacks. The resulting \textit{mutually-authenticated and firmware-attested handshake} runs in the order of $10^2ms$ and $10^1mJ$ on a Cortex-M7 board, making it suitable for integration into routine connection handshakes between IoT devices in critical infrastructure and commercial services. \textbf{Third, for maintenance,} we leverage the deterministic properties of SRAM footprints from digital/physical twin devices to create a reliable, scalable and privacy-preserving post-deployment update system.

\begin{figure}
    \centering
    \includegraphics[width=0.6\linewidth]{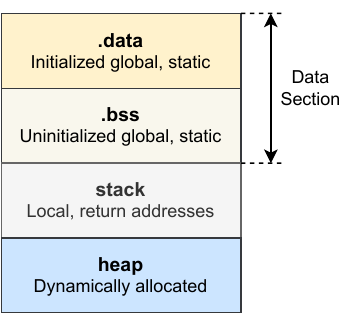}
    \caption{Logical sections of an SRAM.}
    \label{fig:sections}
\end{figure}

\section{Background}
\label{sec:background}

This section presents two empirically-grounded properties of the SRAM that underlie \textit{LiteAtt}'s evidence selection with regard to privacy and ease-of-update.

SRAM is a type of volatile memory that stores runtime data including global and static variables, dynamically allocated variable, and return addresses of functions and interrupts, organized into the logical sections shown in Fig. \ref{fig:sections}. The data section constitutes $.data/.bss$, whose size is determined by the global and static variable dependencies of the firmware, while the $stack$ and $heap$ sections are used dynamically. A firmware utilizes the SRAM according to the variable and function definitions in its code, and consequently leaves a firmware-specific footprint on the SRAM during its runtime. Prior empirical studies on SRAM-based attestation and fingerprinting \cite{aman2022machine, iqbal2024ram, kohli2024swarm, kohli2024safe, holcomb2008power, kohli2024intelligent} support the following two properties, which we adopt as the basis for $\text{App}_{SA}$'s design:

\begin{property}
\label{prop:2}
    Twin devices with the same firmware leave similar patterns in the SRAM data sections \cite{kohli2024safe, kohli2024swarm}.
\end{property}
Let $\mathcal{M}_1, \mathcal{M}_2$ be two physical or digital twin IoT devices (identical design and functionality), and let $\mathcal{F}$ denote the firmware binary loaded onto both. If $\mathcal{F}$ initializes and uses SRAM deterministically, the SRAM contents on $\mathcal{M}_1$ and $\mathcal{M}_2$ exhibit similar patterns.

\begin{property}
\label{prop:3}
    Physical twin devices have distinct initializations of the SRAM \cite{holcomb2008power, kohli2024intelligent}.
\end{property}
Under the same setup as Property \ref{prop:2}, the SRAM contents on $\mathcal{M}_1$ and $\mathcal{M}_2$ exhibit distinct power-up initializations due to minute hardware imperfections. However, updates to the memory by the firmware follow deterministic patterns as stated in Property \ref{prop:2}

\begin{figure}[t]
    \centering
    \includegraphics[width=\linewidth]{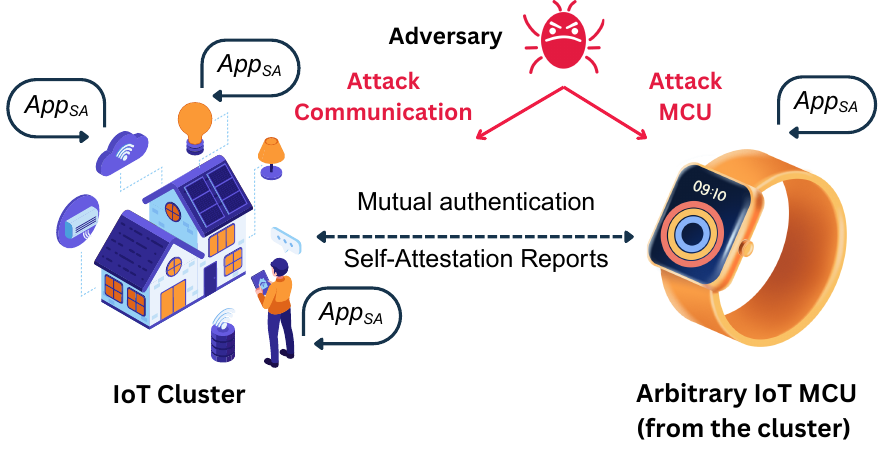}
    \caption{Overview of the IoT network model for SA. Capable MCUs in the cluster furnish SA reports during the connection handshake to enhance trust.}
    \label{fig:network}
\end{figure}

\begin{figure*}[t]
    \centering
    \includegraphics[width=0.85\linewidth]{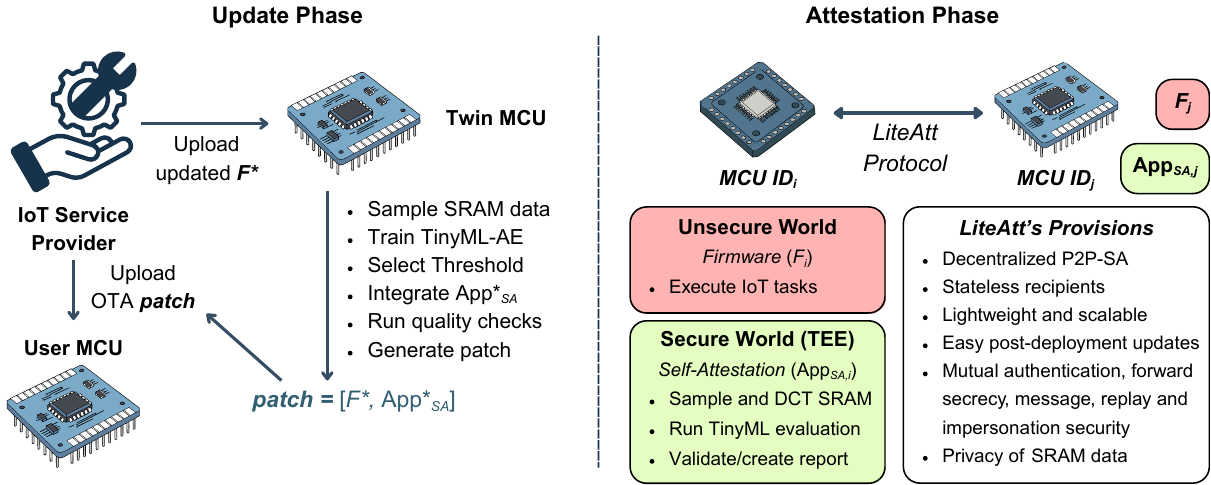}
    \caption{Overview of the proposed \textit{LiteAtt} attestation framework. \textit{LiteAtt} enables trust in identity and firmware, incurs low overheads, ensures SRAM privacy, and enables post-deployment updates.}
    \label{fig:proposed}
\end{figure*}

Property \ref{prop:2} enables IoT vendors to train the TinyML models on data from twin devices, preserving the privacy of user SRAM and enabling ease of post-deployment updates. Property \ref{prop:3} appears to preclude attesting the $stack/heap$ region, since its power-up state differs across twins. However, this device-specific variation stems from independent, cell-level hardware imperfections and is therefore spatially high-frequency, whereas the firmware-determined component of the stack/heap is comparatively low-frequency and structured. \textit{LiteAtt}'s DCT-based compression (Section \ref{sec:proposed}) retains only the low-frequency coefficients, attenuating the device-specific initialization noise while preserving the firmware-determined structure, and thereby extends attestation to both the $.data/.bss$ and $stack/heap$ sections across twins.

\section{Network and Threat Model}
\label{sec:network}
Fig. \ref{fig:network} presents the P2P-SA network model considered in this paper. There are two types of participants in this model, namely, an IoT cluster and an adversary.
\begin{enumerate}
    \item \textit{IoT Cluster ($\mathcal{C}_{IoT}$):} A dynamic cluster of $n$ MCUs, where each MCU ($ID_i, i\in[n]$) is equipped with its application firmware binary ($\mathcal{F}$) and an SA application binary ($\text{App}_{SA}$). We assume that the MCUs can run TinyML inference, and host a TEE for secure execution and storage. Each MCU is provisioned with a long-term secret and public key-pair ($SK_i, PK_i$) and a certificate ($\text{Cert}_i$) signed by the IoT vendor's Certificate Authority (CA), enabling certificate-based authentication. During the connection handshake, two connecting MCUs (say, $ID_i$ and $ID_j$) act as \textit{provers} ($\mathcal{P}_i$ and $\mathcal{P}_j$, respectively). Each device presents a TEE-signed SA report ($\sigma$) with its counterpart to establish trust in its firmware. More details of the proposed P2P-SA algorithm and protocol are presented in Section \ref{sec:proposed}.
    \item \textit{Adversary ($\mathcal{A}$):} A malicious entity that attempts to undermine firmware attestation. We assume that $\mathcal{A}$ can modify or replace $\mathcal{F}$ in the unsecure region of the IoT MCUs to malware ($\mathcal{F}_m$). $\mathcal{A}$ may also target the communication between MCUs, i.e., it may eavesdrop, intercept, modify, replay or drop messages between two MCUs, and may attempt to impersonate either MCU during the handshake. We assume that $\mathcal{A}$ cannot extract or tamper with the TEE contents and is computationally bounded to polynomial time. Side-channel attacks are assumed to be mitigable through constant-time software and split-cache implementations \cite{weissteiner2025teecorrelate} and are not part of the attack scope.
\end{enumerate}

\section{Proposed P2P-SA Framework: \textbf{\textit{LiteAtt}}}
\label{sec:proposed}
This section presents the proposed SA framework called \textit{LiteAtt}, which includes the update (or setup) phase and attestation phase, the SA algorithm ($\text{App}_{SA}$), and the \textit{LiteAtt} protocol. Fig. \ref{fig:proposed} provides an overview of the framework. 

\subsection{Update Phase}

Upon developing the first or updated firmware $\mathcal{F}^*_j$ for $ID_j$, the IoT service provider must correspondingly create the attestation binary $\text{App}_{SA,j}$. This is done as follows:
\begin{enumerate}
    \item Upload $\mathcal{F}^{*}_j$ to $Twin_j$ and collect a data set ($D$) of SRAM snapshots captured during the runtime of $\mathcal{F}^*_j$ on $Twin_j$. As shown in Fig. \ref{fig:sections}, a snapshot comprises a statically-allocated $.data/.bss$ region and a dynamically-allocated $stack/heap$ region, which \textit{LiteAtt} attests independently. Each $L$-byte snapshot $B_0,\dots,B_{L-1}$ is partitioned at the $.data/.bss$ boundary $\lambda$ into a $.data/.bss$ section (bytes $[0,\lambda)$) and a $stack/heap$ section (bytes $[\lambda,L)$). To obtain a fixed-length, low-overhead representation irrespective of a section's byte length $\ell$, each section is compressed using a Type-II Discrete Cosine Transform (DCT), retaining only its $F$ lowest-frequency coefficients:
    \begin{equation}
    \label{eq:dct}
        S_k = \alpha_k\sum_{n=0}^{\ell-1}\frac{B_n}{255}\cos\!\left[\frac{\pi(2n+1)k}{2\ell}\right], \;\; k\in\{0,1,...,F-1\} \; ,
    \end{equation}
    where $B_n$ is the $n^{th}$ byte of the section, and $\alpha_0=\sqrt{1/\ell}$, $\alpha_k=\sqrt{2/\ell}\;(k>0)$ ortho-normalize the transform. The compressed section $S = [S_0,...,S_{F-1}]$ has a fixed dimensionality $F$, which keeps the model input size, and consequently $\text{App}_{SA}$'s memory and compute overhead, constant regardless of the SRAM size. Moreover, a localized firmware modification is broadband in the frequency domain, spreading its energy across many coefficients $S_k$, and thus remains detectable despite the compression (Section \ref{sec:results}). Each coefficient is standardized using the per-feature mean and standard deviation of the training data. The compressed, standardized vectors of both sections constitute $D$, which is randomly divided into $D_{train}$ and $D_{val}$ with a 2:1 ratio. Scaled uniform noise $\epsilon\sim\mathcal{U}(0,1) \in \mathbb{R}^{n\times F}$ is added to $D_{train}$ as follows:
    \begin{equation}
    \tilde{D}_{train} = D_{train} + n_f\cdot\epsilon \; ,
    \end{equation}
    where $n_f$ is the noise scaling factor. Uniform noise helps prevent overfitting through denoising-based training of the TinyML AEs.
    \item Initialize two AEs, $M^A$ and $M^B$, for the $.data/.bss$ and $stack/heap$ sections, respectively, and optimize over the task $\hat{S}=M(\tilde{S}), \forall \tilde{S},S \in \tilde{D}_{train},D_{train}$ such that the mean squared error, $MSE(\hat{S},S)$, is minimized. Convert each $M$ to a quantized TinyML model ($M_{lite}$) using the appropriate conversion library.
    \item Adaptively select the target True Negative Rate ($TNR_{target}$) over $D_{val}$ based on the distribution of reconstruction MSE obtained on $D_{val}$. This is done to balance True Positive Rate ($TPR_{test}$) and $TNR_{test}$ over a test set with \textit{prior knowledge of only normal SRAM patterns}. The gap between the $95^{th}$ and $99^{th}$ percentiles with reference to the $95^{th}$ percentile quantifies the spread of validation errors as follows:
    \begin{equation}
        \Gamma = \frac{P_{99} - P_{95}}{P_{95}} \; ,
    \end{equation}
    where $P_{95}$ and $P_{99}$ denote the $95^{th}$ and $99^{th}$ percentiles of validation error, respectively. We use $P_{95}$ as a conservative lower bound (tighter threshold) of $TNR_{val}$ to filter outliers and ensure potentially high $TPR_{test}$ on the unknown malicious test cases. Conversely, $P_{99}$ is the relaxed upper bound of $TNR_{target}$. $TNR_{target}$ is determined adaptively using a piecewise function as follows:
    \begin{equation}
    \label{eq:adaptive}
    \text{TNR}_{\text{target}} = 
    \begin{cases}
    0.99 & \text{if } \Gamma < 0.2 \\
    0.97 & \text{if } 0.2 \leq \Gamma < 0.5 \\
    0.95 & \text{if } \Gamma \geq 0.5
    \end{cases} \; ,
    \end{equation}
    where a small $\Gamma$ ($< 0.2$) indicates tightly clustered validation errors, suggesting scope to increase $TNR_{test}$ with minimal change to the threshold. Conversely, a large gap ratio ($\geq 0.5$) indicates dispersed errors, suggesting a more conservative $TNR_{val}$ (95\%) to potentially improve $TPR_{test}$. Discrete TNR targets and gap limits are selected heuristically for best performance on all datasets. Given the $TNR_{target}$, we employ binary search to determine the optimal threshold MSE ($\mathcal{T}_{opt}$) such that:
    \begin{equation}
    \label{eq:binary}
    \left| \text{TNR}_{\text{val}}(\mathcal{T}_{opt}) - \text{TNR}_{\text{target}} \right| < 0.005 \; ,
    \end{equation}
    where $\text{TNR}_{\text{val}}(\mathcal{T}_{opt})$ is the empirical $TNR$ achieved on $D_{val}$ using threshold $\mathcal{T}_{opt}$. This procedure is applied independently to each section model. Since the two models are combined disjunctively during attestation (Algorithm \ref{algo:generate}), each is calibrated to a target of $\sqrt{TNR_{target}}$ so that the combined attestation specificity meets $TNR_{target}$.
    \item Evaluate the trained section models and their thresholds on benign and adversarial samples to ensure adversarial resilience (see Section \ref{sec:setup}) and low False Positive Rate (FPR). If cleared, the detection parameters are integrated into $\text{App}^*_{SA,j}$ and a $patch=[\mathcal{F}^*_j, \text{App}^*_{SA}]$ is created for Over-the-Air (OTA) delivery.
    \item Provision $ID_j$ with (i) the $\mathcal{F}^*_j$ and $\text{App}^*_{SA,j}$ binaries, (ii) a long-term key-pair $(SK_j, PK_j)$ generated within the TEE, (iii) a certificate $\text{Cert}_j$ binding $PK_j$ to $ID_j$ and signed by the vendor's CA, and (iv) the CA's public key for verifying certificates of peer devices. $SK_j$, the section models ($M^A_{lite}$, $M^B_{lite}$), and their thresholds ($\mathcal{T}^A_{opt}$, $\mathcal{T}^B_{opt}$) never leave the TEE. All updates are delivered to the user device via OTA updates.
\end{enumerate} 

\begin{algorithm}[t]
\caption{Report generation algorithm: $\text{App}_{SA}\text{.Generate}(ctx_{self})$.}
\label{algo:generate}

\textbf{Inputs:} Session transcript context $ctx$.

\textbf{Output:} SA report $\mathcal{R}_{self}$.

$S^A \leftarrow \textsc{DCT}(B_{0:\lambda})$ \Comment{DCT-compressed .data/.bss}

$\gamma \leftarrow \mathds{1}[MSE(M^A_{lite}(S^A), S^A) > \mathcal{T}^A_{opt}]$ \Comment{Attest data section}

\If{$\gamma==0$}{

    $S^B \leftarrow \textsc{DCT}(B_{\lambda:L})$ \Comment{DCT-compressed $stack/heap$}
    
    $\gamma \leftarrow \mathds{1}[MSE(M^B_{lite}(S^B), S^B) > \mathcal{T}^B_{opt}]$ \Comment{Attest $stack/heap$}
}

$t \leftarrow time()$

$N \leftarrow PRNG()$

$\mathcal{R}_{self} \leftarrow (ID_{self}, \gamma, t, N, ctx_{self})$

return $\mathcal{R}_{self}$

\end{algorithm}

\begin{algorithm}[t]
\caption{Report validation algorithm: $\text{App}_{SA}\text{.Validate}(\mathcal{R}_s, ID_{expected}, ctx_{expected})$.}
\label{algo:validate}

\textbf{Inputs:} Received report $\mathcal{R}_s$, expected sender ID $ID_{expected}$, session transcript context $ctx_{expected}$.

\textbf{Output:} validation outcome $\delta_s \in \{-2, -1, 0, 1\}$

$(ID_s, \gamma_s, t_s, N_s, ctx_s) \leftarrow \mathcal{R}_s$ \Comment{Parse received report}

\If{$ID_s \neq ID_{expected} \: \lor \: ctx_s \neq ctx_{expected}$}{
    return $-2$ \Comment{Identity or session mismatch}
}

\If{$time() - t_s > \epsilon$}{
    return $-1$ \Comment{Expired report}
}

\If{$\gamma_s == 1$}{
    return $1$ \Comment{Sender reports unsafe firmware}
}

return $0$ \Comment{Safe sender}

\end{algorithm}

\begin{figure}
    \centering
    \includegraphics[width=\linewidth]{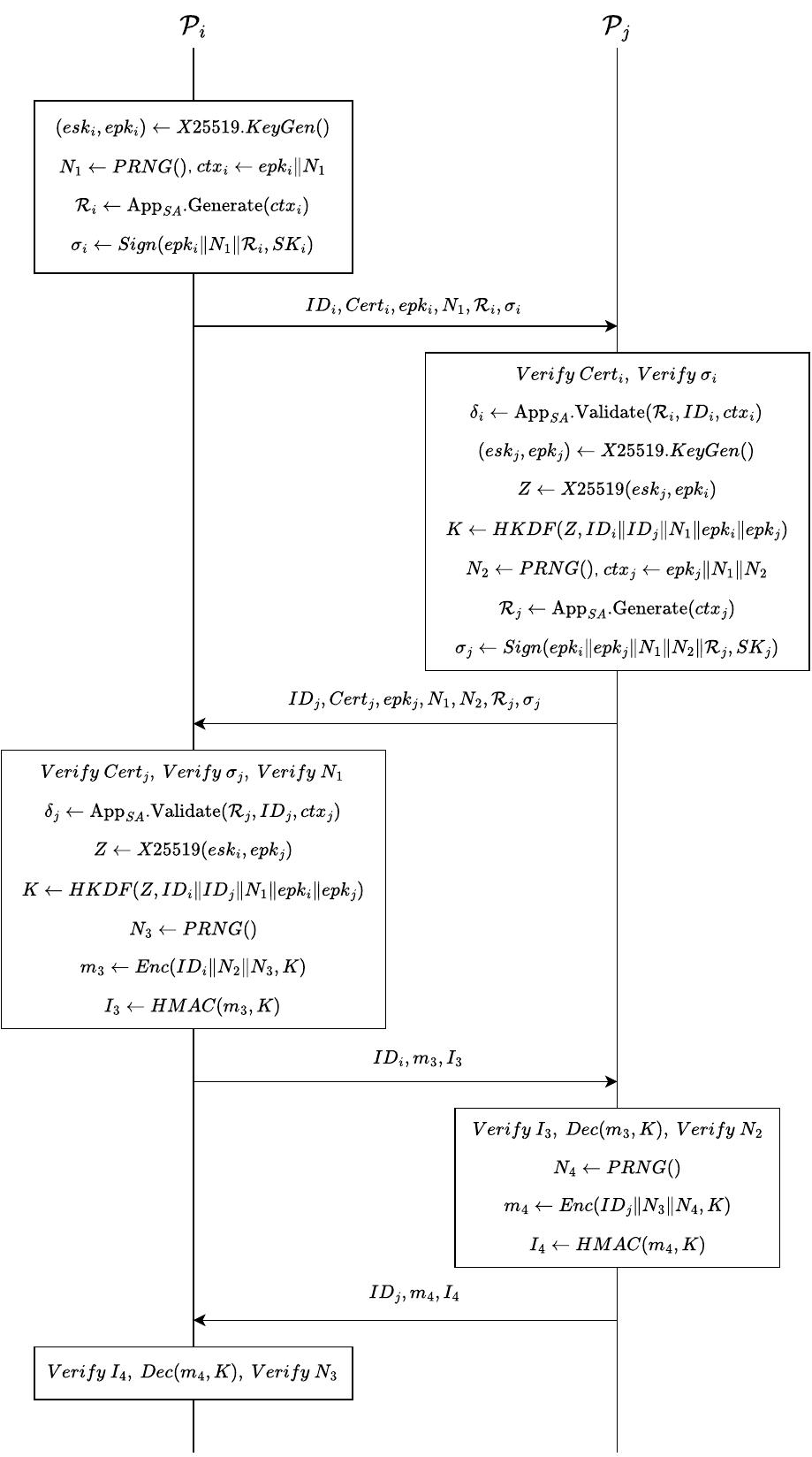}
    \caption{\textit{LiteAtt} protocol for mutual authentication and SA.}
    \label{fig:self-att-protocol}
\end{figure}

\subsection{Attestation Phase}
The attestation phase occurs during the routine operation of the MCUs in $\mathcal{C}_{IoT}$. Device $ID_i$ in the network may attempt to connect with $ID_j$ to share data or control signals. To do so, $ID_i$ and $ID_j$ participate in a four-way handshake protocol, invoking their respective $\text{App}_{SA}$ to prove their identities and firmware integrity before any P2P communication. Devices use X25519 for ephemeral Diffie-Hellman key agreement, Ed25519 for identity signatures ($Sign$, $Verify$), and HMAC-based Key Derivation Function ($HKDF$) for session-key derivation. Once a session key has been established, devices use Advanced Encryption Standard Cipher Block Chaining (AES-CBC) for encryption ($Enc$) and decryption ($Dec$), and Secure Hash Algorithm (SHA-256) for Hash-based Message Authentication Code ($HMAC$). Pseudo-random nonces are generated for freshness using a Pseudo-random Number Generator ($PRNG$), and a synchronized monotonic clock ($time$) is used for timestamps.

\subsubsection{$\text{App}_{SA}$:} $\text{App}_{SA}$ provides two TEE-resident routines, $\text{App}_{SA}.\texttt{Generate}$, and $\text{App}_{SA}.\texttt{Validate}$ for report generation and validation, respectively. Both routines operate entirely within the TEE and never expose intermediate state to the unsecure world. Please note that cryptographic operations are handled by the protocol.

$\text{App}_{SA}.\texttt{Generate}(ctx_{self})$ samples and DCT-compresses the $.data/.bss$ and $stack/heap$ sections from the SRAM, and thresholds the reconstruction error of the corresponding section models ($M^A_{lite}$, $M^B_{lite}$) against $\mathcal{T}^A_{opt}$ and $\mathcal{T}^B_{opt}$. The SA outcome $\gamma$ is their disjunction, with the $stack/heap$ model evaluated only when the data section is clean. It then assembles a report $\mathcal{R}_{self}=[ID_{self},\gamma,t,N,ctx_{self}]$.

$\text{App}_{SA}.\texttt{Validate}(\mathcal{R}_s, ID_{expected}, ctx_{expected})$, on the other hand, parses a received report ($\mathcal{R}_s$), and returns one of four outcomes: $-2$ if either $ID$ or $ctx$ in the payload does not match the expected values, $-1$ if the report is stale ($time() - t > \epsilon$), $1$ if the sender's $\gamma_s = 1$, and $0$ otherwise. The context checks bind the report to the specific session, defeating replay across sessions even within the freshness window. $\epsilon = 1$s is chosen to comfortably cover worst-case report generation, signing, transmission, and verification latency across the target boards while bounding the harvested-replay window.

\subsubsection{P2P-SA Protocol:} Fig. \ref{fig:self-att-protocol} depicts the proposed mutual authentication, key agreement, and attestation protocol between $\mathcal{P}_i$ and $\mathcal{P}_j$. Each device is provisioned with a long-term identity key-pair $(SK, PK)$, an associated certificate $\text{Cert}$ signed by the IoT vendor's Certification Authority (CA), and the CA's public key for verifying peer certificates. One successful run of the protocol comprises the following five steps:

\begin{enumerate}
    \item[\textbf{Step 1.}] $\mathcal{P}_i$ generates a fresh nonce $N_1 \leftarrow PRNG()$ and an ephemeral X25519 key-pair $(esk_i, epk_i)$. It computes the session context $ctx_i = epk_i \| N_1$, invokes $\mathcal{R}_i \leftarrow \text{App}_{SA}.\texttt{Generate}(ctx_i)$, and signs the handshake transcript as $\sigma_i \leftarrow \text{Sign}(epk_i \| N_1 \| \mathcal{R}_i, SK_i)$. It then sends $(ID_i, \text{Cert}_i, epk_i, N_1, \mathcal{R}_i, \sigma_i)$ to $\mathcal{P}_j$.

    \item[\textbf{Step 2.}] $\mathcal{P}_j$ verifies $\text{Cert}_i$ against the CA's public key and $\sigma_i$ using $PK_i$ extracted from $\text{Cert}_i$. It invokes $\delta_i \leftarrow \text{App}_{SA}.\texttt{Validate}(\mathcal{R}_i, ID_i, ctx_i)$, and proceeds only if $\delta_i = 0$. $\mathcal{P}_j$ then generates its own ephemeral key-pair $(esk_j, epk_j)$, computes the shared secret $Z = X25519(esk_j, epk_i)$, derives the session key $K \leftarrow \text{HKDF}(Z, ID_i \| ID_j \| N_1 \| epk_i \| epk_j)$, and generates $N_2 \leftarrow PRNG()$. It then computes $ctx_j = epk_j \| N_1 \| N_2$, invokes $\mathcal{R}_j \leftarrow \text{App}_{SA}.\texttt{Generate}(ctx_j)$, signs the transcript $\sigma_j \leftarrow \text{Sign}(epk_i \| epk_j \| N_1 \| N_2 \| \mathcal{R}_j, SK_j)$, and sends $(ID_j, \text{Cert}_j, epk_j, N_1, N_2, \mathcal{R}_j, \sigma_j)$ to $\mathcal{P}_i$. The inclusion of $N_1$ in the signed transcript binds the response to $\mathcal{P}_i$'s initial request.

    \item[\textbf{Step 3.}] $\mathcal{P}_i$ verifies $\text{Cert}_j$ and $\sigma_j$, and confirms that the returned $N_1$ matches its own. It computes $Z = X25519(esk_i, epk_j)$ and derives the matching session key $K \leftarrow \text{HKDF}(Z, ID_i \| ID_j \| N_1 \| epk_i \| epk_j)$. It invokes $\delta_j \leftarrow \text{App}_{SA}.\texttt{Validate}(\mathcal{R}_j, ID_j, ctx_j)$, and proceeds only if $\delta_j = 0$. $\mathcal{P}_i$ then generates $N_3 \leftarrow PRNG()$, constructs $m_3 \leftarrow Enc(ID_i \| N_2 \| N_3, K)$ and $I_3 \leftarrow HMAC(m_3, K)$ under the newly derived session key, and sends $(ID_i, m_3, I_3)$ to $\mathcal{P}_j$. This message serves as key confirmation since its successful decryption by $\mathcal{P}_j$ proves that $\mathcal{P}_i$ derived the correct session key, and the included $N_2$ acknowledges $\mathcal{P}_j$'s freshness.

    \item[\textbf{Step 4.}] $\mathcal{P}_j$ verifies $I_3$, decrypts $m_3$ using $K$, and confirms that the returned $N_2$ matches its own. This implicitly authenticates $\mathcal{P}_i$'s knowledge of $Z$ and completes mutual key confirmation. $\mathcal{P}_j$ then generates $N_4 \leftarrow PRNG()$, constructs $m_4 \leftarrow Enc(ID_j \| N_3 \| N_4, K)$ and $I_4 \leftarrow HMAC(m_4, K)$, and sends $(ID_j, m_4, I_4)$ to $\mathcal{P}_i$.

    \item[\textbf{Step 5.}] As the final step, $\mathcal{P}_i$ verifies $I_4$, decrypts $m_4$, and confirms that the returned $N_3$ matches its own, completing the protocol. Both parties may now use $K$ for the remainder of the session and discard the ephemeral keys $esk_i, esk_j$.
\end{enumerate}

A thorough formal security evaluation is presented in Section~\ref{sec:security}. In summary, the protocol guarantees cryptographically strong mutual authentication through certificate-bound identity signatures, alongside post-handshake message confidentiality and integrity driven by the derived session key $K$. Because $K$ is computed exclusively from ephemeral X25519 key-pairs that are immediately destroyed upon session termination, the architecture inherently ensures perfect forward secrecy. While $\mathcal{R}_i, \mathcal{R}_j$ are transmitted unencrypted, their authenticity and immutable session-binding are strictly enforced by $\sigma$ and $ctx$ checks executed inside the TEE. We explicitly do not provide confidentiality for $\gamma$ since the binary attestation outcome is already inferable from handshake completion. Finally, defense against replay is maintained globally via the integration of fresh 128-bit pseudo-random nonces across all four handshake segments, whereas $\epsilon$-comparison prevents the injection of stale or harvested attestation reports.

\begin{figure}
    \centering
    \includegraphics[width=0.8\linewidth]{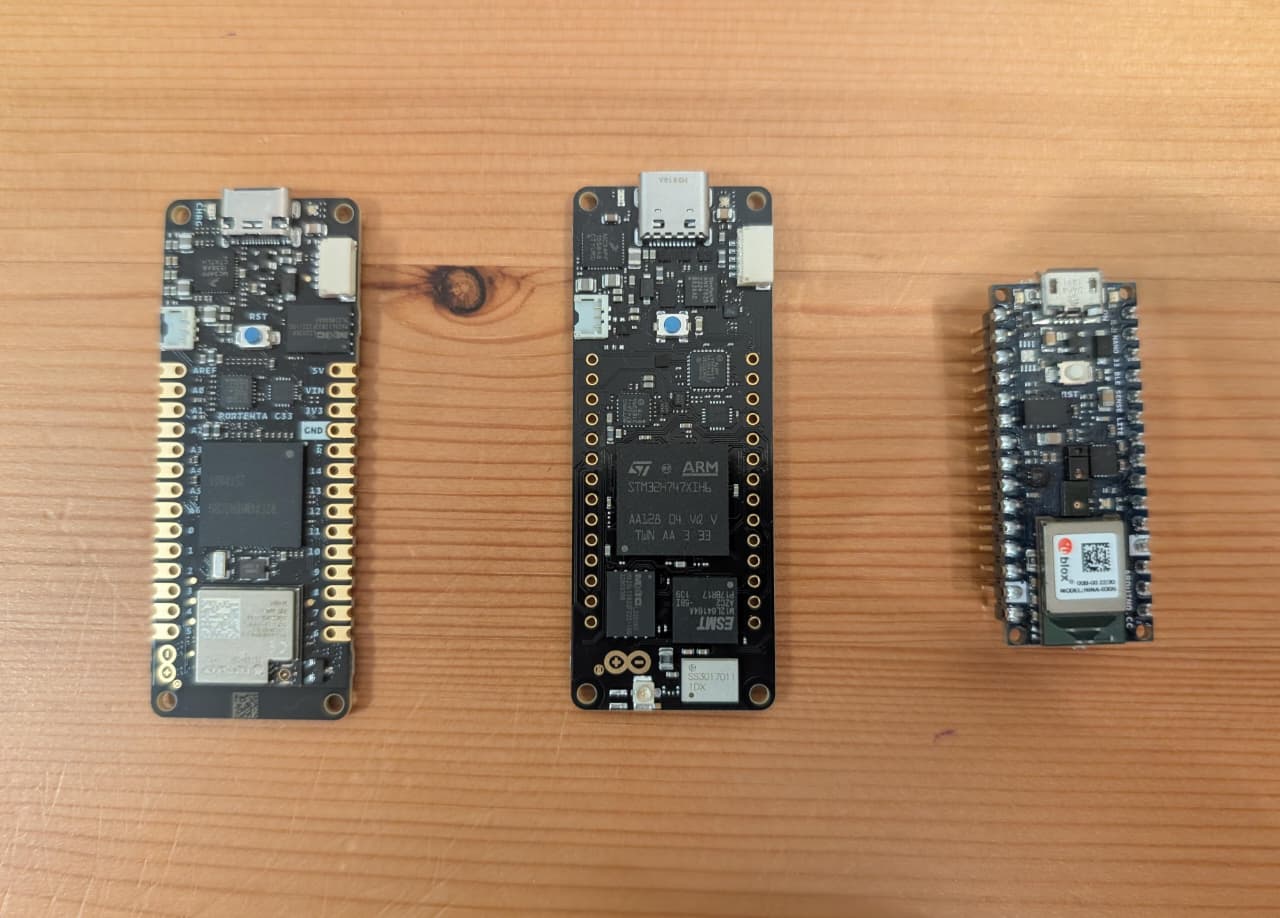}
    \caption{(Left to right) Arduino Portenta C33, Portenta H7, and Nano 33 BLE Sense MCU boards used to evaluate \textit{LiteAtt}.}
    \label{fig:hardware}
\end{figure}

\begin{table}[t]
\centering
\caption{Hardware specifications of the Arm Cortex-M boards used in this paper.}
\label{tab:hardware}
\renewcommand{\arraystretch}{1.2}
\resizebox{\columnwidth}{!}{
\begin{tabular}{|ll|l|l|l|}
\hline
\multicolumn{2}{|l|}{\textbf{Property}}                                             & \textbf{Portenta C33} & \textbf{Portenta H7}                                                          & \textbf{Nano 33 BLE Sense} \\ \hline
\multicolumn{2}{|l|}{Processor}                                                         & Cortex-M33    & Cortex-M7, -M4 & Cortex-M4F          \\ \hline
\multicolumn{2}{|l|}{Clock speed} & 200 MHz & 480, 240 MHz & 64 MHz \\ \hline
\multicolumn{2}{|l|}{Core voltage} & 1.1 V  & 1.2 V & 1.8 V \\ \hline
\multicolumn{2}{|l|}{Active current} & 30 mA  & 280 mA & 5 mA \\ \hline
\multicolumn{2}{|l|}{Active power} & 33 mW  & 336 mW & 9 mW \\ \hline
\multicolumn{1}{|l|}{\multirow{2}{*}{RAM}}                    & SRAM                    & 512KB                & 1MB                                                                          & 256KB                     \\ \cline{2-5} 
\multicolumn{1}{|l|}{}                                        & DRAM                    & -                     & 8MB                                                                          & -                          \\ \hline
\multicolumn{2}{|l|}{Flash}                                                             & 18 (2+16)MB          & 18 (2+16)MB                                                                  & 1MB                       \\ \hline
\multicolumn{2}{|l|}{TinyML}                                                            & Yes                   & Yes                                                                           & Yes                        \\ \hline
\multicolumn{2}{|l|}{TEE}                                                               & Yes                   & No (dual core)                                                                & No                         \\ \hline
\multicolumn{2}{|l|}{\begin{tabular}[c]{@{}l@{}}Cryptographic\\ primitive\end{tabular}} & SE050C2               & \begin{tabular}[c]{@{}l@{}}SE050C2\\ ATECC608\end{tabular}                    & ATECC608A                  \\ \hline
\end{tabular}}
\end{table}

\section{Experimental Setup}
This section presents the experimental testbed, including hardware, software, dataset, attacks, TinyML hyperparameters, and evaluation metrics.

\label{sec:setup}
\subsection{Hardware and Software}
\subsubsection{Hardware} We use three cost-effective, TinyML-compatible Arm Cortex-M Arduino boards with a range of capabilities to evaluate the efficacy of \textit{LiteAtt}. These include the TEE-enabled Arduino Portenta C33 \cite{portena-c33}, the dual-core Arduino Portenta H7 \cite{portena-h7}, and the cheapest among the three, the Arduino Nano 33 BLE Sense \cite{nano-33}. Fig. \ref{fig:hardware} shows the physical devices while Table \ref{tab:hardware} compiles the corresponding hardware specifications. The Arduino Portenta C33 is equipped with the 200 MHz Arm Cortex-M33 core, 2MB flash + 16MB external flash, 512KB SRAM, the SE050C2 secure element, and features an Arm Trustzone TEE, making it our main test bed. The Arduino Portenta H7 is equipped with the 480 MHz Arm Cortex-M7 and 240 MHz Cortex-M4 dual cores, 2MB flash + 16MB external flash, and 1MB SRAM + 8MB external SDRAM. The Arduino Nano 33 BLE is equipped with the 64 MHz Arm Cortex-M4F core, 1MB flash, 256KB SRAM, and the ATECC608A cryptoprocessor.

\subsubsection{Software} TinyML models are designed using Tensorflow 2.17 and Python 3.11. The models are quantized to 8-bit integer parameters using TFLite Micro. The protocol's cryptographic operations and $\text{App}_{SA}$ are created in Renesas $e^2$ Studio for Arm TrustZone on the Arduino Portenta C33, and Arduino IDE for Arduino Portenta H7 and Nano 33 BLE Sense. Software-based cryptographic primitives were used for consistency across the boards. 

\begin{table}[t]
\centering
\caption{Datasets. The S/M column gives the number of \textit{safe} (S) and \textit{unsafe} (U) variants of each firmware.}
\label{tab:datasets}
\renewcommand{\arraystretch}{1.2}
\resizebox{\columnwidth}{!}{
\begin{tabular}{|ccll|c|}
\hline
\multicolumn{1}{|c|}{\textbf{Dataset}} & \multicolumn{1}{l|}{\textbf{Type}} & \multicolumn{1}{l|}{\textbf{Firmware}} & \textbf{Brief description} & \textbf{S/U} \\ \hline
\multicolumn{1}{|c|}{\multirow{8}{*}{1 \cite{0nze-r023-24}}} & \multicolumn{1}{l|}{\multirow{8}{*}{Single}} & \multicolumn{1}{l|}{AES128} & Performs looped encryption, decryption & 1/3 \\ \cline{3-5}
\multicolumn{1}{|c|}{} & \multicolumn{1}{l|}{} & \multicolumn{1}{l|}{Interrupt} & Push-button interrupt & 1/3 \\ \cline{3-5}
\multicolumn{1}{|c|}{} & \multicolumn{1}{l|}{} & \multicolumn{1}{l|}{LED} & Control LED using analog pin readings & 1/3 \\ \cline{3-5}
\multicolumn{1}{|c|}{} & \multicolumn{1}{l|}{} & \multicolumn{1}{l|}{Random} & Generates pseudo-random numbers & 1/3 \\ \cline{3-5}
\multicolumn{1}{|c|}{} & \multicolumn{1}{l|}{} & \multicolumn{1}{l|}{Shake} & Detects lateral movement & 1/3 \\ \cline{3-5}
\multicolumn{1}{|c|}{} & \multicolumn{1}{l|}{} & \multicolumn{1}{l|}{Temperature} & Reads surrounding temperature & 1/3 \\ \cline{3-5}
\multicolumn{1}{|c|}{} & \multicolumn{1}{l|}{} & \multicolumn{1}{l|}{Vibration} & Detects vibration, controls LED & 1/3 \\ \cline{3-5}
\multicolumn{1}{|c|}{} & \multicolumn{1}{l|}{} & \multicolumn{1}{l|}{XTS} & AES-XTS block cipher & 1/3 \\ \hline
\multicolumn{1}{|c|}{\multirow{4}{*}{2 \cite{gmee-vj41-24}}} & \multicolumn{1}{l|}{\multirow{4}{*}{Swarm}} & \multicolumn{1}{l|}{N0\_4} & Master node of a four-node swarm & 1/1 \\ \cline{3-5}
\multicolumn{1}{|c|}{} & \multicolumn{1}{l|}{} & \multicolumn{1}{l|}{N1\_4} & Senses 24 bytes of data for processor node & 1/1 \\ \cline{3-5}
\multicolumn{1}{|c|}{} & \multicolumn{1}{l|}{} & \multicolumn{1}{l|}{N2\_4} & Processes data, generates a control signal & 1/2 \\ \cline{3-5}
\multicolumn{1}{|c|}{} & \multicolumn{1}{l|}{} & \multicolumn{1}{l|}{N3\_4} & Control peripherals using the control signal & 1/3 \\ \hline
\multicolumn{1}{|c|}{\multirow{6}{*}{3 \cite{gmee-vj41-24}}} & \multicolumn{1}{l|}{\multirow{6}{*}{Swarm}} & \multicolumn{1}{l|}{N0\_6} & Master node of a six-node swarm & 1/1 \\ \cline{3-5}
\multicolumn{1}{|c|}{} & \multicolumn{1}{l|}{} & \multicolumn{1}{l|}{N1\_6} & Senses 16 bytes of data for processor node & 1/1 \\ \cline{3-5}
\multicolumn{1}{|c|}{} & \multicolumn{1}{l|}{} & \multicolumn{1}{l|}{N2\_6} & Processes data, generates a control signal & 1/2 \\ \cline{3-5}
\multicolumn{1}{|c|}{} & \multicolumn{1}{l|}{} & \multicolumn{1}{l|}{N3\_6} & Controls peripheral using the control signal & 1/3 \\ \cline{3-5}
\multicolumn{1}{|c|}{} & \multicolumn{1}{l|}{} & \multicolumn{1}{l|}{N4\_6} & Senses 12 bytes of data for processor node & 1/1 \\ \cline{3-5}
\multicolumn{1}{|c|}{} & \multicolumn{1}{l|}{} & \multicolumn{1}{l|}{N5\_6} & Processes data, controls peripherals & 1/2 \\ \hline
\multicolumn{1}{|c|}{\multirow{5}{*}{4}} & \multicolumn{1}{l|}{\multirow{5}{*}{Single}} & \multicolumn{1}{l|}{Environment} & Temperature, humidity, light monitoring & 1/- \\ \cline{3-5}
\multicolumn{1}{|c|}{} & \multicolumn{1}{l|}{} & \multicolumn{1}{l|}{Network} & Packet routing, checksum verification & 1/- \\ \cline{3-5}
\multicolumn{1}{|c|}{} & \multicolumn{1}{l|}{} & \multicolumn{1}{l|}{Motor} & Dual motor control with safety monitoring & 1/- \\ \cline{3-5}
\multicolumn{1}{|c|}{} & \multicolumn{1}{l|}{} & \multicolumn{1}{l|}{Security} & Token-based HMAC, authentication & 1/- \\ \cline{3-5}
\multicolumn{1}{|c|}{} & \multicolumn{1}{l|}{} & \multicolumn{1}{l|}{Data} & Statistical analytics of multi-sensor data & 1/- \\ \hline
\multicolumn{4}{|l|}{\textbf{Total}} & \textbf{23/41} \\ \hline
\end{tabular}}
\end{table}

\subsection{Firmware, Datasets and Attacks}
Table \ref{tab:datasets} compiles the firmware samples used in this paper. Dataset 1 comprises SRAM dumps from 8 firmwares, including cryptographic, sensing, and control applications, and \textit{24 modifications to data and functional dependencies}. Datasets 2 and 3 SRAM dumps from swarm deployments, including a variety of sense, process, and peripheral control tasks with three to twenty-four bytes of collected, processed, and communicated data. This dataset comprises of \textit{10 modified samples with data and functional modifications,} and \textit{7 data injection cases arising from modified firmware in swarm settings.}. Dataset 4 is collected to further aid the study with more complex sensing, control, and analytics applications.

All SRAM traces comprise unsigned integer byte values in the range 0-255, partitioned into $.data/.bss$ and $stack/heap$ sections that are each compressed to $F$ DCT coefficients via Equation (\ref{eq:dct}). The \textit{safe} samples of every $\mathcal{F}_i$ follow a (50\%, 25\%, 25\%) split for training, validation, and testing. The validation set is used to select the adaptive detection threshold, $\mathcal{T}_{opt}$. Only \textit{safe} samples are used to train the $M_{lite}$, consistent with the anomaly-detection formulation.

\textit{\textbf{Test-time scenarios.}} We evaluate the detector against five classes of unsafe SRAM state, listed below. The first two are drawn from the collected datasets, while the last three are synthetically simulated on held-out safe traces to characterize $M_{lite}$ under controlled perturbation.
\begin{itemize}
    \item \textbf{\textit{Unsafe:}} For firmware $\mathcal{F}_i$ with adversarial variants in Datasets 1-3, the malicious samples of $\mathcal{F}_i$ (with changes to data and function dependencies, and data injection attacks) constitute its \textit{unsafe} set. Discriminating between \textit{safe} and \textit{unsafe} categories is a challenging task since modifications and data injection cause a minor change in the SRAM footprint of the respective firmware. Each \textit{safe} firmware has upto 3 \textit{unsafe} firmware counterparts, as seen in Table \ref{tab:datasets}.
    \item \textbf{\textit{Other:}} The cross-firmware unsafe set of SRAM samples from every $\mathcal{F}_{\backslash i}$ in the dataset to evaluate discrimination between different firmware. Each \textit{safe} firmware has 22 \textit{other} counterparts.
    \item \textbf{\textit{DOP:}} Targeted variable bytes within the $.data/.bss$ section are synthetically corrupted across a controlled byte extent. This simulates DOP runtime modifications that bypass static binary code-integrity checks by altering only runtime data. 
    \item \textbf{\textit{ROP/JOP:}} A contiguous chain of valid 32-bit flash memory code-addresses (representing gadget sequences) is injected into aligned boundaries of the $stack/heap$ section. JOP is specifically modeled by additionally inserting a recurring dispatcher address profile. These simulate the memory footprint of active execution-hijacking payloads that traditional CFA schemes \cite{ammar2025sok} target through heavy hardware instrumentation.
    \item \textbf{\textit{Random-perturbation sensitivity:}} A random perturbation adversarial analysis instantiating a weak adversary $\mathcal{A}_{obliv}$ adversary aiming to fool $M_{lite}$. Safe traces are corrupted at random by sampling $\eta\%$ of bytes uniformly in each SRAM section and replacing them with uniformly random byte values, sweeping $\eta\in\{0.1, 0.25, 0.5, 1, 2, 5, 10, 25, 50\}\%$. This is the structurally weakest perturbation an adversary with a $\eta$-byte budget can apply.
    \item \textbf{\textit{Adaptive attack:}} A decision-based black-box adversarial-ML analysis instantiating $\mathcal{A}_{adapt}$ with access to two authorized devices under its control (one of which hosts the target $M_{lite}$) with a goal to make adaptive perturbations using $\mathcal{F}_m$ and that fool $M_{lite}$. We assume an offline twin-query budget $Q{=}\infty$ between the captured devices, and sweep across perturbations $\eta\in\{1,2,5,10\}\%$ for three adaptive tiers: \textit{(i)}~$\mathcal{A}_{adapt}^{none}$ does not have access to benign SRAM dumps and samples $Q=5000$ random $L_0$ perturbations of the malicious trace, querying the twin-oracle, observing $\gamma$ feedback, and replaying the highest-scoring candidates against the target model. \textit{(ii)}~$\mathcal{A}_{adapt}^{xfw}$ holds benign SRAM from \emph{other} firmware, trains a generator model ($\widehat{M}$) on about 1,700 cross-firmware traces, and applies candidate perturbations onto each malicious trace to fool the target $M_{lite}$. \textit{(iii)}~$\mathcal{A}_{adapt}^{1}$ knows exactly one leaked benign trace $x_b$ of the target firmware. It samples $Q$ candidates by noisy top-$\eta\%$ disagreement substitution from $x_b$, queries the twin to label them, trains $\widehat{M}$ on the twin-confirmed positives, and submits the highest-ranked candidates against the target $M_{lite}$. This is the strongest adversarial attack among the three but assumes access to benign SRAM. Each tier yields the empirical bound $\text{FNR}_{\text{adapt}}^{tier}(\eta)$ used in Section \ref{sec:security}.
\end{itemize}

Each TinyAE is trained on $\sim$900 \textit{safe} samples only, and evaluated using an average 450 \textit{safe}, 60,000 \textit{unsafe}/\textit{other}, and 155,000 DOP/ROP/JOP samples.

\subsection{ML and Hyperparameters}
We consider three simple TinyAE architectures for attestation and one MLP-VAE architecture for adversarial attacks:
\begin{enumerate}
    \item $M_1$: A two-layer Multi-layer Perceptron (MLP)-AE architecture with $l$ input features, eight hidden neurons followed by a 0.2 dropout, and $l$-neuron linear output.
    \item $M_2$: A three-layer MLP-AE architecture with $l$-input features, two hidden layers of eight hidden neurons each, followed by a 0.2 dropout, and a $l$-neuron linear output.
    \item $M_3$: A four-layer Convolutional Neural Network (CNN)-AE architecture with $l$-input features, two convolutional encoding layers with sixteen and eight filters of three dimensions, each followed by a 2-dimensional maxpooling layer, an eight-neuron hidden layer followed by a dropout, and a $l$-neuron linear output.
    \item $\widehat{M}$: A 4-layer MLP-VAE with $L$ inputs, 256 dimension encoder and decoder, and 16$\times$2 dimension bottleneck.
\end{enumerate}

All layers except the output are activated using Rectified Linear Unit (ReLU) activation. $M_1-M_3$ are trained using the Adam optimizer with a batch size of 64 and a learning rate of 0.005 for 100 epochs, while $\widehat{M}$ is trained using the Adapm optimizer with a batch size of 32, 0.001 learning rate and KL weight, and 40 epochs.

\subsection{Evaluation Metrics}
\textit{LiteAtt}'s predictive performance is evaluated using standard ML metrics including Accuracy, Precision, TPR, TNR, FPR, False Negative Rate (FNR), and F1-score. Discriminative ability is also evaluated using the Receiver Operating Characteristic-Area Under the Curve (ROC-AUC). Further, runtime performance of the protocol includes peak memory overhead (in KB), attestation latency (in ms), and energy consumption (in $\mu$J).

\begin{table}[t]
\centering
\caption{Comparison of memory overheads and key predictive performance statistics on the Arduino Portenta C33 for model selection.}
\label{tab:models}
\renewcommand{\arraystretch}{1.2}
\resizebox{\columnwidth}{!}{
\begin{tabular}{|l|l|l|l|}
\hline
\multicolumn{1}{|c|}{\multirow{2}{*}{\textbf{Property}}} & \multicolumn{1}{c|}{$M_1$} & \multicolumn{1}{c|}{$M_2$} & \multicolumn{1}{c|}{$M_3$} \\ \cline{2-4} 
\multicolumn{1}{|c|}{}                                   & \textit{2-layer MLP-AE}          & \textit{3-layer MLP-AE}          & \textit{4-layer CNN-AE}          \\ \hline
Keras size  & 37.64KB          & 43.10KB          & 68.49KB          \\ \hline
TFlite size & 4.91KB           & 5.66KB           & 11.89KB          \\ \hline
Reduction    & $7.66\times$     & $7.61\times$     & $5.76\times$     \\ \hline
Tensor arena & 1.91KB           & 2.02KB           & 8.80KB           \\ \hline
Accuracy          & 0.9942 ± 0.0122  & 0.9940 ± 0.0122  & 0.9941 ± 0.0121  \\ \hline
TPR               & 0.9945 ± 0.0122  & 0.9943 ± 0.0123  & 0.9944 ± 0.0122  \\ \hline
TNR               & 0.9514 ± 0.0223  & 0.9489 ± 0.0212  & 0.9479 ± 0.0216  \\ \hline
\end{tabular}}
\end{table}

\begin{figure}[t]
    \centering
    \includegraphics[width=\linewidth]{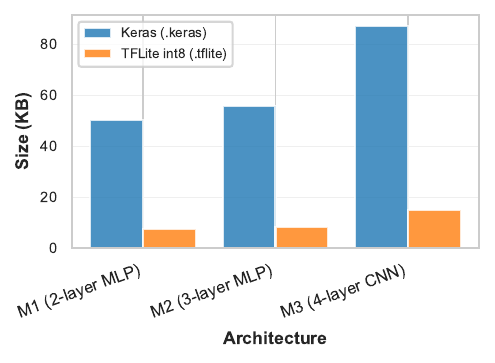}
    \caption{Keras vs TFLite model size comparison.}
    \label{fig:size}
\end{figure}

\begin{figure}[t]
    \centering
    \includegraphics[width=\linewidth]{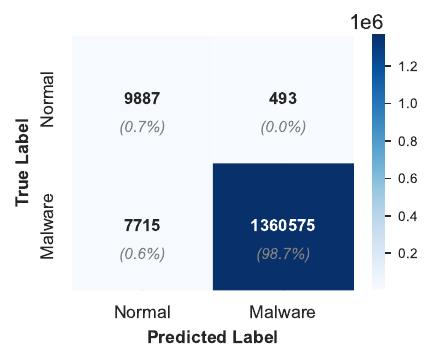}
    \caption{Combined confusion matrix across all firmware, detectors, and test samples.}
    \label{fig:confusion}
\end{figure}

\begin{figure}[t]
    \centering
    \includegraphics[width=\linewidth]{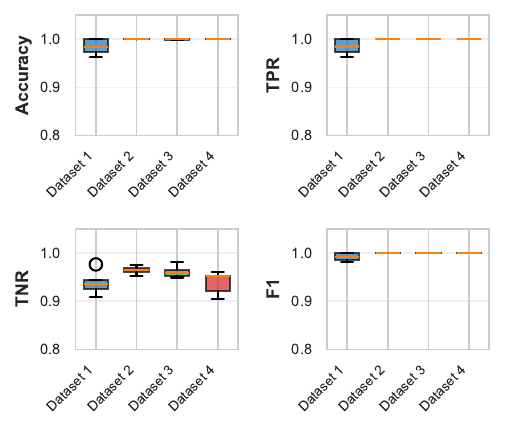}
    \caption{Comparison of key metrics across datasets.}
    \label{fig:datasets}
\end{figure}

\begin{figure}[t]
    \centering
    \includegraphics[width=\linewidth]{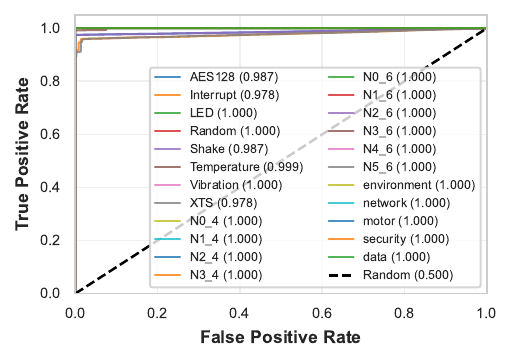}
    \caption{ROC-AUC for all firmware.}
    \label{fig:rocauc}
\end{figure}

\begin{table*}[t]
\centering
\caption{Performance metrics using $M_1$ for all firmware.}
\label{tab:prediction}
\renewcommand{\arraystretch}{1.2}
\resizebox{0.9\textwidth}{!}{
\begin{tabular}{|c|c|c|c|c|c|c|c|c|c|}
\hline
\textbf{Firmware} & $TNR_{target}$ & \textbf{A} & \textbf{P} & \textbf{Overall TPR} & \textbf{Unsafe TPR} & \textbf{Other TPR} & \textbf{FPR} & \textbf{FNR} & \textbf{F1} \\ \hline
AES128      & 0.9500 & 0.9760 & 0.9996 & 0.9763 & 1.0000 & 0.9756 & 0.0587 & 0.0237 & 0.9878 \\ \hline
Interrupt   & 0.9500 & 0.9621 & 0.9995 & 0.9624 & 0.8272 & 0.9683 & 0.0800 & 0.0376 & 0.9806 \\ \hline
LED         & 0.9500 & 0.9996 & 0.9996 & 1.0000 & 1.0000 & 1.0000 & 0.0667 & 0.0000 & 0.9998 \\ \hline
Random      & 0.9500 & 0.9999 & 0.9999 & 1.0000 & 1.0000 & 1.0000 & 0.0400 & 0.0000 & 1.0000 \\ \hline
Shake       & 0.9500 & 0.9763 & 0.9996 & 0.9765 & 1.0000 & 0.9759 & 0.0613 & 0.0235 & 0.9879 \\ \hline
Temperature & 0.9500 & 0.9954 & 0.9996 & 0.9957 & 0.8287 & 1.0000 & 0.0587 & 0.0043 & 0.9977 \\ \hline
Vibration   & 0.9500 & 0.9996 & 0.9996 & 1.0000 & 1.0000 & 1.0000 & 0.0613 & 0.0000 & 0.9998 \\ \hline
XTS         & 0.9500 & 0.9619 & 0.9995 & 0.9622 & 0.8264 & 0.9681 & 0.0827 & 0.0378 & 0.9805 \\ \hline
N0\_4       & 0.9700 & 0.9998 & 0.9998 & 1.0000 & 1.0000 & 1.0000 & 0.0275 & 0.0000 & 0.9999 \\ \hline
N1\_4       & 0.9700 & 0.9998 & 0.9998 & 1.0000 & 1.0000 & 1.0000 & 0.0275 & 0.0000 & 0.9999 \\ \hline
N2\_4       & 0.9700 & 0.9998 & 0.9998 & 1.0000 & 1.0000 & 1.0000 & 0.0275 & 0.0000 & 0.9999 \\ \hline
N3\_4       & 0.9700 & 0.9998 & 0.9998 & 1.0000 & 1.0000 & 1.0000 & 0.0250 & 0.0000 & 0.9999 \\ \hline
N0\_6       & 0.9500 & 0.9996 & 0.9996 & 1.0000 & 1.0000 & 1.0000 & 0.0278 & 0.0000 & 0.9998 \\ \hline
N1\_6       & 0.9700 & 0.9996 & 0.9996 & 1.0000 & 1.0000 & 1.0000 & 0.0233 & 0.0000 & 0.9998 \\ \hline
N2\_6       & 0.9700 & 0.9994 & 0.9994 & 1.0000 & 1.0000 & 1.0000 & 0.0389 & 0.0000 & 0.9997 \\ \hline
N3\_6       & 0.9700 & 0.9995 & 0.9995 & 1.0000 & 1.0000 & 1.0000 & 0.0311 & 0.0000 & 0.9998 \\ \hline
N4\_6       & 0.9500 & 0.9995 & 0.9995 & 1.0000 & 1.0000 & 1.0000 & 0.0311 & 0.0000 & 0.9998 \\ \hline
N5\_6       & 0.9700 & 0.9996 & 0.9996 & 1.0000 & 1.0000 & 1.0000 & 0.0233 & 0.0000 & 0.9998 \\ \hline
Environment & 0.9500 & 0.9999 & 0.9999 & 1.0000 & --     & 1.0000 & 0.0317 & 0.0000 & 1.0000 \\ \hline
Network     & 0.9500 & 0.9998 & 0.9998 & 1.0000 & --     & 1.0000 & 0.0873 & 0.0000 & 0.9999 \\ \hline
Motor       & 0.9500 & 0.9999 & 0.9999 & 1.0000 & --     & 1.0000 & 0.0635 & 0.0000 & 0.9999 \\ \hline
Security    & 0.9500 & 0.9999 & 0.9999 & 1.0000 & --     & 1.0000 & 0.0556 & 0.0000 & 0.9999 \\ \hline
Data        & 0.9500 & 0.9998 & 0.9998 & 1.0000 & --     & 1.0000 & 0.0873 & 0.0000 & 0.9999 \\ \hline
\textbf{Average} & \textbf{0.9570} & \textbf{0.9942} & \textbf{0.9997} & \textbf{0.9945} & \textbf{0.9712} & \textbf{0.9951} & \textbf{0.0486} & \textbf{0.0055} & \textbf{0.9970} \\ \hline
\end{tabular}}
\end{table*}

\begin{figure*}[t]
    \centering
    \includegraphics[width=0.9\linewidth]{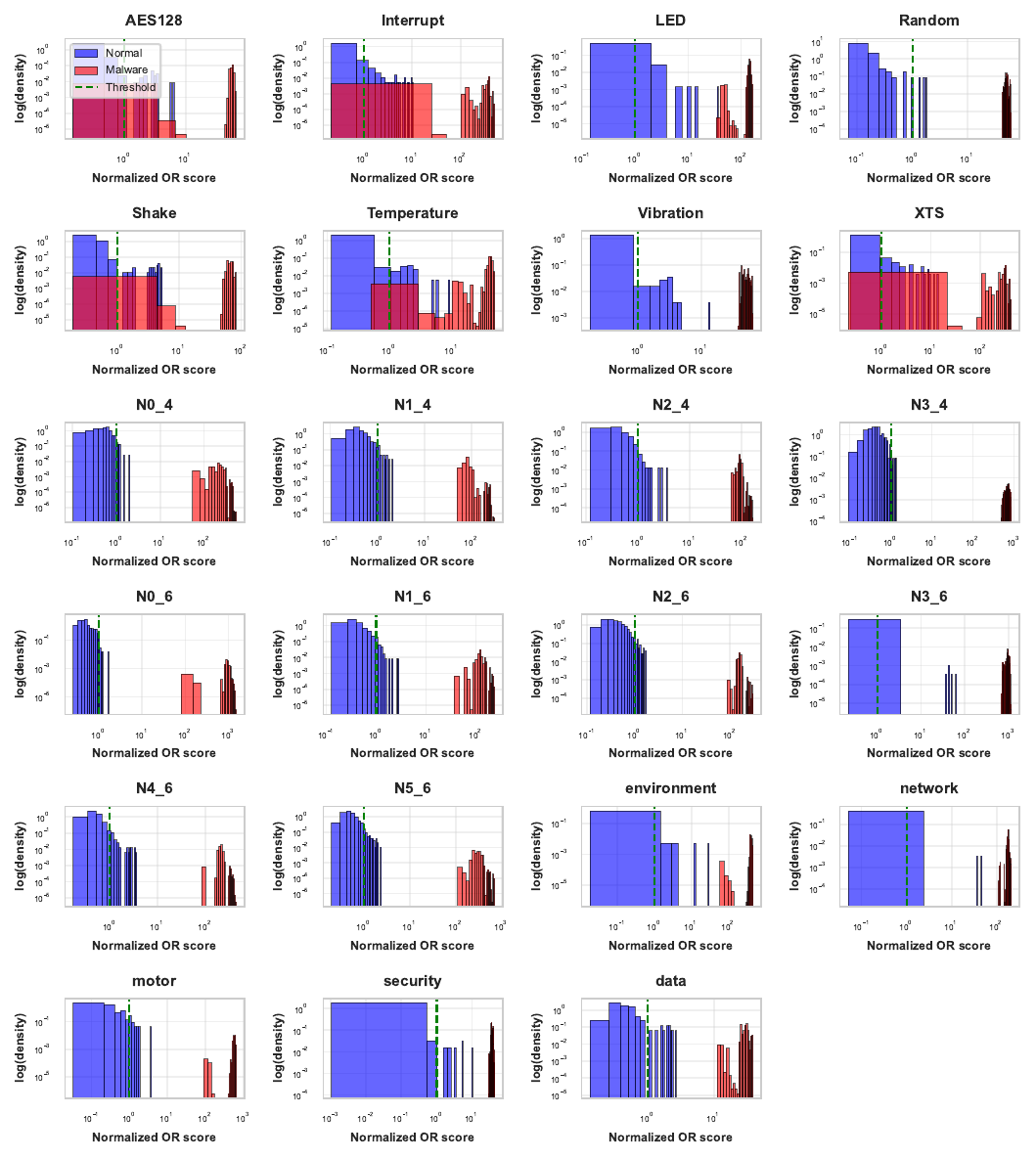}
    \caption{Firmware-wise logarithmic reconstruction error distribution with respect to thresholds.}
    \label{fig:error}
\end{figure*}

\section{Results}
\label{sec:results}
This section presents the experimental results for predictive performance on modified or injected cases and simulated attacks, evaluation of physical twin data, sensitivity and scalability analyses, and comparison with prior work.

\subsection{TinyML Model Performance}
We set $\epsilon=1$s as a conservative freshness period, $n_f=0.2$ for denoising, and $F=64$. Table \ref{tab:models} compares the predictive performance and memory overheads of three candidate architectures, $M_1$, $M_2$, and $M_3$, while Figure \ref{fig:size} demonstrates the .keras to .tflite size reduction. All models were trained on the same SRAM data section traces from \textit{safe} firmware. $M_1$ achieves the best balance between predictive accuracy and resource efficiency, achieving an accuracy of 99.42\% (± 1.22\%), a TPR of 99.45\% (± 1.22\%), and the highest TNR of 95.14\% (± 2.23\%) with just 4.9KB models and 1.91KB tensor arena. $M_2$ matches $M_1$ on accuracy (99.40\%) and TPR (99.43\%) with a marginally lower TNR (94.89\%), but at the cost of increased memory. Finally, $M_3$ achieves comparable accuracy (99.41\%) but requires a significantly larger model and tensor arena size, due to its convolutional feature maps. This is also the reason for a lower size reduction from .tflite conversion in $M_2$ ($5.76\times$) compared to $M_1$ ($7.66\times$) and $M_2$ ($7.61\times$). These results indicate that the added architectural complexity of $M_2$ and $M_3$ does not result in proportional improvement, while substantially increasing the runtime memory footprint. Based on this analysis, $M_1$ is selected for all subsequent experiments.

Table \ref{tab:prediction} presents the firmware-wise performance of $M_1$ across all SRAM datasets. $M_1$ achieves an average accuracy of 99.42\%, precision of 99.97\%, TPR of 99.45\%, TNR of 95.14\%, and F1-score of 99.70\%. The overall confusion matrix for all anomaly detection tasks is shown in Fig. \ref{fig:confusion} and the dataset-wise performance is highlighted in Fig. \ref{fig:datasets}. $M_1$ achieves 97.12\% TPR on \textit{unsafe} samples, showing discriminative capability for data and function modifications in the binary, and data injection attacks using the DCT-compressed SRAM $.data/.bss$ and $stack/heap$ sections. Further, the model achieves 99.51\% TPR on \textit{other} samples. Fig. \ref{fig:rocauc} presents the ROC-AUC for all firmware. \textit{LiteAtt} achieves an ROC-AUC of 100\% for 18 out of 23 firmware and exceeds 97.8\% for all firmware, highlighting strong discrimination in anomaly detection. Fig. \ref{fig:error} illustrates separability between the reconstruction error of normal and malware classes relative to the selected conservative thresholds, indicating that the adaptive thresholding mechanism in Equations (\ref{eq:adaptive}) and (\ref{eq:binary}) conservatively adjusts the TNR target based on the spread of validation errors to achieve high TPR across all firmware.

\subsection{Performance on Simulated Attacks}
\label{sec:sim-attacks}
Fig.~\ref{fig:sens-extent} presents DOP, ROP, and JOP detection by $M_1$. ROP and JOP are detected at $\geq$98\% from 16~bytes onward, and DOP from 64~bytes. All three saturate near 100\% beyond 5\% of the section. ROP/JOP are easier than DOP because their 32-bit flash-address patterns are highly non-uniform. In contract DOP perturbations are small and localized, overlapping with benign low-frequency variation and yielding a weaker DCT signal.

In addition, Fig.~\ref{fig:sens-random} presents the random-perturbation sensitivity analysis across all firmware, which also serves as the quality analysis in the update phase (Section \ref{sec:proposed}). As the figure shows, mean detection crosses 95\% at $\eta\!\approx\!2\%$ and saturates to 100\% above $\eta=5\%$ for the $.data/.bss$ section models, while the $stack/heap$ models reach 95\% at $\eta\!\approx\!1\%$. Fig.~\ref{fig:sens-adapt-vs-obliv} contrasts the random-perturbation bound against the three sub-tiers of $\mathcal{A}_{adapt}$. Both $\mathcal{A}_{adapt}^{none}$ and $\mathcal{A}_{adapt}^{xfw}$ yield $FNR\!=\!0$ at every $\eta$ and only $\mathcal{A}_{adapt}^{1}$ registers non-zero advantage ranging from $0.077$ at $\eta\!=\!1\%$ to $0.287$ at $\eta\!=\!10\%$. The qualitative jump between $\mathcal{A}_{adapt}^{xfw}$ and $\mathcal{A}_{adapt}^{1}$ identifies benign-SRAM extraction as the adversarial capability that breaks the defense. These curves form the empirical basis for the attestation bounds shown in in Section~\ref{sec:security}.

\begin{figure}[t]
    \centering
    \includegraphics[width=\linewidth]{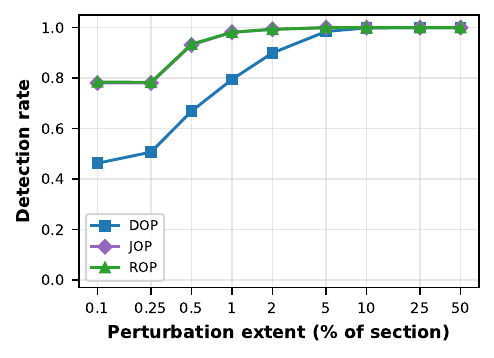}
    \caption{Detection rate vs.\ perturbation extent for simulated DOP/ROP/JOP attacks on $M_1$ at $F=64$.}
    \label{fig:sens-extent}
\end{figure}

\begin{figure}[t]
    \centering
    \includegraphics[width=\linewidth]{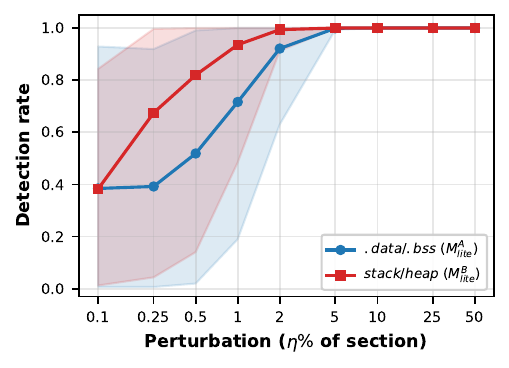}
    \caption{TPR of $M_1$ ($F{=}64$) on \textit{safe} traces with $\eta\%$ of bytes replaced by uniformly random values. Shaded bands represent per-firmware min-max.}
    \label{fig:sens-random}
\end{figure}

\begin{figure}[t]
    \centering
    \includegraphics[width=\linewidth]{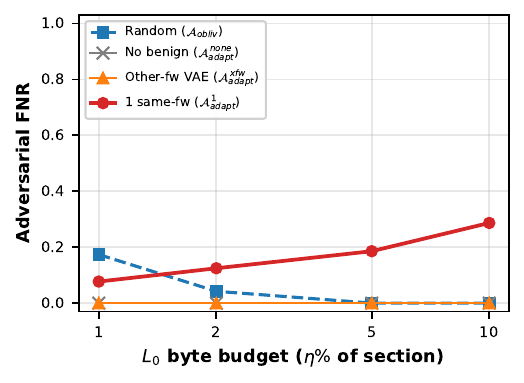}
    \caption{Adversarial FNR vs $\eta$ for four types of adversarial attacks on $M_{lite}$.}
    \label{fig:sens-adapt-vs-obliv}
\end{figure}

\subsection{Performance on Twin Data}
In addition to distinguishing between various firmwares, their data and functional modifications, data injection, DOP, ROP, and JOP attacks using anomaly detection, we also demonstrate the transferability of TinyML models trained on twin DCT-compressed SRAM data to support Property \ref{prop:2} while mitigating the impact of Property \ref{prop:3}. $M_1$ achieves a 97.2\% TNR and 100\% TPR (nearly equal to the values reported in Table \ref{tab:prediction}) when trained on twin samples and tested on safe samples using the adaptive threshold. This outcome applies to both $M^A_{lite}$ and $M^B_{lite}$ and highlights the transferability of DCT-compressed SRAM traces between twin devices hosting the same firmware despite Property \ref{prop:3}, supporting the use of twin devices to create \textit{LiteAtt} updates post-deployments.

\begin{table}[t]
\centering
\caption{Per-operation latency and energy overheads.}
\label{tab:runtime}
\renewcommand{\arraystretch}{1.2}
\resizebox{\columnwidth}{!}{
\begin{tabular}{|l|c|cc|cc|cc|}
\hline
\multirow{2}{*}{\textbf{Operation}} & \multirow{2}{*}{\textbf{Count}} & \multicolumn{2}{c|}{\textbf{Portenta C33}} & \multicolumn{2}{c|}{\textbf{Portenta H7}} & \multicolumn{2}{c|}{\textbf{Nano 33 BLE}} \\ \cline{3-8}
 & & Lat (ms) & E ($\mu$J) & Lat (ms) & E ($\mu$J) & Lat (ms) & E ($\mu$J) \\ \hline
Harvest + DCT  & 2 & 29.226 & 964.5  & 5.631 & 1892.1 & 50.478 & 454.3 \\ \hline
Standardize & 2 & 0.014 & 0.47   & 0.007 & 2.3     & 0.059  & 0.5 \\ \hline
int8 inference & 2 & 0.129 & 4.26   & 0.034 & 11.6    & 0.386  & 3.5 \\ \hline
Threshold  & 2 & 0.006 & 0.21   & 0.001 & 0.45    & 0.009  & 0.08 \\ \hline
Payload                     & 1 & 0.003 & 0.09   & 0.002 & 0.77    & 0.017  & 0.15 \\ \hline
Ed25519 sign & 1 & 19.06 & 629.1  & 1.44  & 485.5   & 18.19  & 163.7 \\ \hline
X25519 keygen               & 1 & 42.41 & 1399.5 & 2.81  & 942.6   & 36.56  & 329.0 \\ \hline
X25519 DH                   & 1 & 42.41 & 1399.5 & 2.81  & 944.0   & 36.59  & 329.3 \\ \hline
HKDF                        & 1 & 0.291 & 9.6    & 0.054 & 18.1    & 0.709  & 6.4 \\ \hline
Ed25519 verify & 2 & 59.88 & 1975.9 & 3.89  & 1305.5 & 49.82  & 448.3 \\ \hline
AES-CBC & 2 & 0.156 & 5.14   & 0.024 & 8.15    & 0.325  & 2.9 \\ \hline
HMAC-SHA256 & 2 & 0.144 & 4.76   & 0.029 & 9.7     & 0.363  & 3.3 \\ \hline
\textbf{Total}              &   & \textbf{283.3} & \textbf{9348.3} & \textbf{26.3} & \textbf{8850.6} & \textbf{294.9} & \textbf{2654.3} \\ \hline
\textbf{Peak RAM} & & \multicolumn{6}{c|}{\textbf{11.83KB}} \\ \hline
\end{tabular}}
\end{table}

\subsection{Overheads}
\subsubsection{Memory}
For $M_1$, each section model occupies 4.91KB as a quantized TFLite binary, which is a $7.66\times$ reduction from the 37.64KB as a Keras model, and uses a 1.91KB tensor arena. The fixed $F=64$-dimensional input keeps the model architecture and tensor sizes constant across firmware and devices. The peak protocol and algorithm footprint shown in Table \ref{tab:runtime} is well within the 256KB-1MB SRAM budgets of the target Arm Cortex-M boards.

\subsubsection{Latency and Energy}
Table \ref{tab:runtime} reports the per-operation latency and energy of one full mutual attestation and handshake on the three boards, measured with the DWT cycle counter and a datasheet-derived energy model ($P_{active} = V_{core} \times I_{active}$). The \textit{Count} column gives per-device invocations of each operation, and totals apply these multipliers and time both section models unconditionally to reflect the worst case. For example, as per Algorithm \ref{algo:generate}, the deployment-time disjunction does not execute Model B when Model A flags the SRAM $.data/.bss$ state. The end-to-end latency per-peer is 283.3ms on the Portenta C33, 26.3ms on the Portenta H7, and 294.9ms on the Nano 33 BLE Sense, total energy is 9.34mJ, 8.85mJ, and 2.65mJ, respectively. Peak $\text{App}_{SA}$ RAM is 11.83KB on all three boards. Please note that the overheads shown here are derived from software implementations of cryptographic primitives. Using hardware acceleration would result in smaller overheads.

To contextualize these low-energy footprints within operational lifetimes, a typical edge node operating on a 3.7V, 2000 mAh LiPo battery has a total energy reservoir of $3.7 V \times 2Ah \times 3,600 s = 26,640J$. Under continuous execution, this capacity theoretically translates to approximately 2.85 million, 3.01 million, and 10.04 million full mutual attestation sessions on the Portenta C33, Portenta H7, and Nano 33 BLE configurations, respectively. This structural efficiency highlights \textit{LiteAtt}'s applicability to long-term installations without risking accelerated power consumption due to attestation overheads.

\subsection{Scalability Analysis}
We now evaluate the scalability of \textit{LiteAtt} to large and dynamic IoT networks with frequent firmware updates. Since attestation is performed on the prover itself and the peer only validates a certificate-bound signature and binary verdict, the recipient requires no per-device ML model or reference hash state. A firmware update requires only an OTA patch to the individual device, with no re-distribution across the network. This means the \textit{\textbf{per-connection attestation cost is constant regardless of network size}}, and network-wide firmware integrity emerges naturally from the composition of independent pairwise handshakes without any topology management overhead. This is an improvement over all prior work on SA and P2P attestation.

\subsection{Comparison with Prior Work}
Table \ref{tab:comp} compares \textit{LiteAtt} against representative ML-based RA, SA, and P2P-SA schemes. Flash-based SA and P2P schemes, including ERASMUS, SIMPLE, FlashAttest, and LIRA-V, are blind to runtime attacks and impose flash-size-dependent prover overheads of the order $\sim10^{-1}$\,s/KB to $\sim10^{1}$,s, precluding seamless handshake integration. Prior SRAM-based RA methods \cite{aman2022machine,iqbal2024ram,kohli2024swarm,kohli2024safe} achieve comparable accuracy by relying on remote verifiers and transmit raw SRAM dumps to a recipient-hosted ML model, sacrificing privacy and scalability. TRACES \cite{caulfield2024traces} incurs continuous instrumentation overhead, and RAGE \cite{chilese2024one} requires CFG instrumentation and ML model maintenance on a remote verifier. TLS-Att \cite{tschofenig2025using} also requires recipients to host up-to-date reference states for comparison. \textit{LiteAtt} is the only scheme that simultaneously provides wide runtime attack coverage, SRAM privacy, a bounded prover-side latency of $10^{-2}$-$10^{-1}$s across real Cortex-M hardware, scalable P2P integration through low-latency handshakes and stateless recipients, and targeted security updates to only the device receiving a firmware update instead of all peers.

\begin{table}[t]
\centering
\caption{Comparison of \textit{LiteAtt} with recent ML-based RA, SA, and P2P-SA approaches in terms of attack coverage, accuracy (if applicable), privacy preservation, P2P scalability, and the order of prover latency.}
\label{tab:comp}
\renewcommand{\arraystretch}{1.2}
\resizebox{\columnwidth}{!}{%
\begin{tabular}{|l|c|c|c|c|c|}
\hline
\textbf{Reference} & \textbf{Attacks} & \textbf{Accuracy} & \textbf{Privacy} & \textbf{P2P \& Scalability} & \textbf{\specialcell{Prover \\ Latency (s)}} \\ \hline
RAGE \cite{chilese2024one}                      & Wide   & 91--98\%         & \cmark & Very poor (ML)   & N/R \\ \hline
TRACES \cite{caulfield2024traces}               & Wide   & N/R              & \cmark & Poor (hash)  & Continuous \\ \hline
Aman et al. \cite{aman2022machine}              & Wide   & 96.0\%           & \xmark & Very poor (ML)   & N/R \\ \hline
Iqbal et al. \cite{iqbal2024ram}                & Wide   & 100.0\%          & \xmark & Very poor (ML)   & N/R \\ \hline
Swarm-Net \cite{kohli2024swarm}                 & Data   & 99.7\%           & \xmark & Very poor (ML)   & N/R \\ \hline
SAFE-IoT \cite{kohli2024safe}                   & Data   & 95.0\%           & \xmark & Very poor (ML)   & N/R \\ \hline
ERASMUS \cite{carpent2018erasmus}               & Static & N/A              & \cmark & Poor (hash)  & $\sim10^{-1}$/KB \\ \hline
SIMPLE \cite{ammar2020simple}                   & Static & N/A              & \cmark & Poor (hash)  & $\sim10^{-1}$/KB \\ \hline
FlashAttest \cite{zhang2025flashattest}         & Static & N/A              & \cmark & Poor (hash)  & $\sim10^{-1}$/KB \\ \hline
SAFEHIVE \cite{ferro2024safehive}               & Static & N/A              & \cmark & Moderate (DHT)  & N/R \\ \hline
LIRA-V \cite{shepherd2021lira}                  & Static & N/A              & \cmark & Poor (hash)  & $\sim10^1$ \\ \hline
TLS-Att \cite{tschofenig2025using}              & \specialcell{Evidence\\dependent} & N/A              & \cmark & Poor (reference)  & \specialcell{Evidence\\dependent} \\ \hline
\textbf{\textit{LiteAtt} (ours)}                & \textbf{Wide} & \textbf{99.42\%} & \cmark & \textbf{Good (PKI only)} & \textbf{$10^{-1}$} \\ \hline
\end{tabular}}
\end{table}

\section{Security Analysis}
\label{sec:security}
This section provides a game-based security analysis of \textit{LiteAtt} following standard formulations for authenticated key-exchange and protocol security. Table \ref{tab:security} summarizes its security provisions.

\textit{\textbf{Security Assumptions:}} We make the following assumptions:
\begin{enumerate}
    \item[$A_1$ (IND-CPA)] AES-CBC provides indistinguishability under chosen plaintext attacks:
    \begin{equation}
        Adv^{IND-CPA}_{AES}(\mathcal{A}) \leq negl(\lambda).
    \end{equation}
    \item[$A_2$ (EUF-CMA)] HMAC-SHA256 and Ed25519 are existentially unforgeable under chosen-message attack:
    \begin{equation}
    Adv^{EUF-CMA}_{HMAC,Ed25519}(\mathcal{A}) \leq negl(\lambda).
    \end{equation}
    \item[$A_3$ (PRG)] $PRNG$ generates unpredictable 128-bit outputs and uniformly random X25519 ephemeral scalars. The probability of a nonce collision across $q$ sessions is bounded by $q^2 \cdot 2^{-128}$.
    \item[$A_4$ (TEE Integrity)] The Arm TrustZone TEE provides (i) isolated execution such that $\text{App}_{SA}$ and cryptographic operations cannot be observed or modified by the unssecure world (ii) secure storage for keys and parameters during runtime. Side-channel attacks are out of scope and assumed mitigated through constant-time software and split-cache implementations \cite{weissteiner2025teecorrelate}.
    \item[$A_5$ (Long-term Key Security)] Each device's long-term identity private key $SK$ is non-extractable from the TEE during operation. The CA's signing key is held offline and out of the adversary's reach. All cryptographic operations execute via secure primitives without key information leaving the protected boundary.
    \item[$A_6$ (SRAM Patterns)] Physical or digital twin devices running the same firmware produce consistent SRAM patterns.
    \item[$A_7$ (Loose Clock Synchronization)] Devices in the same network maintain monotonic clocks loosely synchronized to within a tolerance smaller than the freshness window $\epsilon$.
    \item[$A_8$ (CDH)] X25519 provides computational Diffie-Hellman hardness:
    \begin{equation}
    Adv^{CDH}_{X25519}(\mathcal{A}) \leq negl(\lambda).
    \end{equation}
    \item[$A_9$ (KDF)] HKDF behaves as a secure key derivation function. When invoked on a high-entropy input $Z$ (i.e., the X25519 shared secret), its output $K$ is computationally indistinguishable from a uniformly random string of the same length.
    \item[$A_{10}$ (Flash Atomicity)] Firmware update requires a device reset, which resets all active sessions and ephemeral key material. 
\end{enumerate}

\begin{table}[t]
\centering
\caption{Summary of security properties.}
\label{tab:security}
\renewcommand{\arraystretch}{1.4}
\resizebox{\columnwidth}{!}{
\begin{tabular}{|c|c|c|c|c|}
\hline
\textbf{Property} & \textbf{Game} & \textbf{\begin{tabular}[c]{@{}c@{}}Supporting\\ Assumptions\end{tabular}} & \textbf{\begin{tabular}[c]{@{}c@{}}Adversary\\ Advantage\end{tabular}} & \textbf{Security} \\ \hline
Authentication  & $G_{AUTH}$    & $A_{2,5,8}$ & $negl(\lambda)$ & Cryptographic \\ \hline
Replay          & $G_{REPLAY}$  & $A_{1-3}, A_5, A_7, A_8$ & $q^2 \cdot 2^{-128}$ & Cryptographic \\ \hline
Forward Secrecy & $G_{FS}$      & $A_{3,4,8,9}$ & $negl(\lambda)$ & Cryptographic \\ \hline
Confidentiality & $G_{CONF}$    & $A_{1,2,5,8,9}$ & $negl(\lambda)$ & Cryptographic \\ \hline
Integrity       & $G_{INT}$     & $A_{2,5}$   & $negl(\lambda)$ & Cryptographic \\ \hline
Attestation     & $G_{IMP}$     & $A_{2,4-7}$ & $FNR_{\text{adv}}(\eta)$ & Statistical \\ \hline
SRAM privacy    & $G_{PRIV}$    & $A_{4-6}$ & $negl(\lambda)$ & Cryptographic \\ \hline
\end{tabular}}
\end{table}

\textit{\textbf{Threat Model:}} The adversary $\mathcal{A}$ is a probabilistic polynomial-time (PPT) attacker that can (i) replace $\mathcal{F} \rightarrow \mathcal{F}_m$ in the unsecure world of any prover, (ii) eavesdrop on, modify, inject, replay, or drop network messages, (iii) query polynomial-bounded encryption, MAC, and signing oracles for messages of its choice without learning the underlying keys, (iv) compromise the long-term identity private keys $SK$ of any device \emph{after} a target session has concluded and its ephemeral keys $esk$ have been discarded, and (v) capture one or more authorized devices for offline adversarial modeling. $\mathcal{A}$ cannot (i) extract long-term or ephemeral keys from the TEE during an active session, (ii) tamper with TEE-resident code or data, (iii) compromise the IoT vendor's CA signing key, or (iv) break the underlying cryptographic primitives (AES-CBC, HMAC-SHA256, X25519, Ed25519, HKDF) in polynomial time. We additionally assume $\mathcal{A}$ has no \textit{direct} query access to $M_{lite}$.

We model each security property as a game between a challenger $\mathcal{C}$ and adversary $\mathcal{A}$.

\subsection{Mutual Authentication}
Mutual authentication ensures that both peers verify each other's identity with freshness.

\subsubsection{Entity Authentication} (Game $G_{AUTH}$) $\mathcal{C}$ initializes $\mathcal{P}_i, \mathcal{P}_j$ with long-term identity key-pairs $(SK_i, PK_i), (SK_j, PK_j)$ and certificates $\text{Cert}_i, \text{Cert}_j$. $\mathcal{A}$ has access to the network and to signing oracles $\mathcal{O}_{Sign}(\cdot)$ for messages of its choice under each long-term private key. $\mathcal{C}$ runs polynomial-number of protocol sessions between $\mathcal{P}_i, \mathcal{P}_j$. $\mathcal{A}$ wins if any honest party accepts a session that has no matching conversation at the intended peer.

\begin{theorem}
\label{thm:auth}
Under $(A_2, A_5, A_8)$, $Adv^{AUTH}_{LiteAtt}(\mathcal{A}) \leq negl(\lambda)$.
\end{theorem}
\begin{proof}
For $\mathcal{P}_j$ to accept Step 1 as originating from $\mathcal{P}_i$, the message must include a valid Ed25519 transcript signature $\sigma_i \leftarrow \text{Sign}(epk_i \| N_1 \| \mathcal{R}_i, SK_i)$. Forging such a signature over the current transcript without $SK_i$ is the EUF-CMA forgery experiment, bounded by $negl(\lambda)$ under $A_2$. Non-extractability of $SK_i$ from the TEE ($A_5$) ensures the adversary's only path is explicit signature forgery. The reciprocal case for $\mathcal{P}_j$'s authentication follows symmetrically from $\sigma_j$. Key confirmation in Steps 3-4 additionally requires the adversary to derive the session key $K$, which by $A_8$ requires solving CDH on the X25519 public transcript components $(epk_i, epk_j)$, which is bounded by $negl(\lambda)$.
\end{proof}

\subsubsection{Replay Resistance} (Game $G_{REPLAY}$) $\mathcal{C}$ runs $q$ protocol sessions. $\mathcal{A}$ records the transcripts $\mathcal{H}$ and attempts replay attacks in a fresh session $q+1$. $\mathcal{A}$ wins if any replayed message beyond Step 1 is accepted, or if a replayed report $\mathcal{R}$ from outside the freshness window is accepted.

\begin{theorem}
\label{thm:replay}
Under $(A_{1-3}, A_5, A_7, A_8)$, $Adv^{REPLAY}_{LiteAtt}(\mathcal{A}) \leq q^2 \cdot 2^{-128}$.
\end{theorem}
\begin{proof}
Each session uses fresh 128-bit nonces $N_1, \ldots, N_4$ from $PRNG$ ($A_3$) and a freshly-derived session key $K$ specific to that session's X25519 transcript exchange ($A_8$). A replayed handshake message at Step 2 carries a transcript signature bound to the current session's $N_1$ and $epk_i$, which the adversary cannot reproduce without signature forgery ($A_2$). Replayed encrypted messages $m_k$ ($k \geq 3$) are bound via encryption to nonces from the current session under the current session key, which the adversary cannot produce without solving CDH or breaking AES IND-CPA ($A_1$). 

For the self-attestation reports, embedding a fresh nonce $N$ and a TEE-generated timestamp $t$ ensures freshness. Reports that are stale beyond $\epsilon$ are structurally rejected by $\text{App}_{SA}.\texttt{Validate}$ under the loose clock synchronization of $A_7$. Furthermore, because the current session transcript context $ctx$ is explicitly embedded within the report body $\mathcal{R}$ during generation and checked disjunctively inside $\text{App}_{SA}.\texttt{Validate}$, an old report replayed into a new session will produce a context mismatch ($ctx_s \neq ctx$), causing Algorithm \ref{algo:validate} to return $-2$ and terminate the handshake. The only residual replay scenario is a pure token or nonce collision across $q$ sessions, bounded by the birthday bound $q^2 \cdot 2^{-128}$.
\end{proof}

\subsubsection{Forward Secrecy} (Game $G_{FS}$) The adversary $\mathcal{A}$ observes a completed session transcript $T$ between $\mathcal{P}_i$ and $\mathcal{P}_j$. After the session terminates and the ephemeral keys $esk_i, esk_j$ are discarded, $\mathcal{A}$ obtains both long-term keys $SK_i, SK_j$. $\mathcal{A}$ wins if it can recover the session key $K$ or any encrypted content from $T$.

\begin{theorem}
\label{thm:fs}
Under $(A_3, A_4, A_8, A_9)$, $Adv^{FS}_{LiteAtt}(\mathcal{A}) \leq negl(\lambda)$.
\end{theorem}
\begin{proof}
The session key $K$ is derived as $\text{HKDF}(Z, \cdot)$ where $Z = X25519(esk_i, epk_j) = X25519(esk_j, epk_i)$. The ephemeral scalars $esk_i, esk_j$ are sampled uniformly at random inside the TEE ($A_3$), never leave the TEE boundary ($A_4$), and are securely purged upon session termination. Only the ephemeral public keys $epk_i, epk_j$ appear exposed on the network channel. Recovering $Z$ from $(epk_i, epk_j)$ alone is the structural X25519 CDH problem, bounded by $negl(\lambda)$ under $A_8$. Recovering $K$ given knowledge of $Z$ would require breaking HKDF, bounded by $negl(\lambda)$ under $A_9$. Long-term identity key compromise is therefore mathematically decoupled from the ephemeral session parameters.
\end{proof}

\subsection{Message Security}
We compose AES-CBC encryption with HMAC-SHA256 in an Encrypt-then-MAC (EtM) construction over Steps 3-4 of the protocol. This composition yields IND-CCA2 security from IND-CPA encryption and EUF-CMA MAC. Steps 1-2 are authenticated via transcript signatures (Theorem \ref{thm:auth}).

\subsubsection{Confidentiality} (Game $G_{CONF}$) $\mathcal{A}$ observes network traffic, has access to $\mathcal{O}_{Enc}(\cdot)$ under an unknown key, and submits two equal-length plaintexts $m_a, m_b$ to $\mathcal{C}$. $\mathcal{C}$ selects $x \in \{a,b\}$ uniformly and returns the EtM ciphertext $c_x$. $\mathcal{A}$ outputs $x' \in \{a, b\}$.

\begin{theorem}
\label{thm:conf}
Under $(A_1, A_2, A_5, A_8, A_9)$, $Adv^{CONF}_{LiteAtt}(\mathcal{A}) = \left|Pr[x'=x] - \tfrac{1}{2}\right| \leq negl(\lambda)$.
\end{theorem}
\begin{proof}
This is directly defended from IND-CCA2 security of EtM, applied to transport messages $m_3, m_4$ under the freshly-derived $K$. The session key $K$ remains hidden under CDH ($A_8$), HKDF indistinguishability ($A_9$), and TEE protection ($A_4, A_5$). Confidentiality applies to post-handshake transport content and attestation reports themselves carry only the binary outcome $\gamma$ and freshness metadata, which are not considered confidential in our threat model.
\end{proof}

\subsubsection{Integrity} (Game $G_{INT}$) $\mathcal{C}$ provides $\mathcal{A}$ with $(ID, m, I)$ for $I = HMAC(m, K)$. $\mathcal{A}$ wins if it outputs $(m', I') \neq (m, I)$ that is accepted as authentic.

\begin{theorem}
\label{thm:int}
Under $(A_2, A_5)$, $Adv^{INT}_{LiteAtt}(\mathcal{A}) \leq negl(\lambda)$.
\end{theorem}
\begin{proof}
Acceptance of $(m', I')$ requires $I' = HMAC(m', K)$. Producing such a tag without $K$ is the EUF-CMA forgery experiment, which under $A_2$ has advantage at most $negl(\lambda)$. The session key $K$ is itself protected by $A_4, A_5$ and the CDH hardness of its derivation by $A_8$. The integrity of handshake messages in Steps 1-2 is protected by Ed25519 transcript signatures under $A_2$.
\end{proof}

\subsection{Attestation}
Attestation comprises the integrity of $\text{App}_{SA}$ execution and the privacy of SRAM contents.

\subsubsection{Firmware Impersonation} (Game $\mathcal{G}_{IMP}$) $\mathcal{A}$ controls $\mathcal{P}_i$'s unsecure world and may replace $\mathcal{F} \rightarrow \mathcal{F}_m$ at any time. $\mathcal{A}$ wins if $\mathcal{P}_j$ accepts a session with $\mathcal{P}_i$ in which $\mathcal{P}_i$'s actual executing firmware at the time of acceptance is $\mathcal{F}_m$. We partition $\mathcal{A}$'s capabilities into two classes to obtain meaningful bounds: $\mathcal{A}_{obliv}$ has no query access to $M_{lite}$ and applies perturbations chosen independently of the model. In contrast, $\mathcal{A}_{adapt}$ has offline access to two authorized devices, one of which is the same-firmware twin, and may train a generator $\widetilde{M}$ using its twin-oracle to optimize perturbations on the SRAM to fool $M_{lite}$. We further sub-partition $\mathcal{A}_{adapt}$ by the firmware-specific information it holds about the target's benign SRAM manifold. $\mathcal{A}_{adapt}^{none}$ has twin-oracle access but no benign SRAM of any firmware, while $\mathcal{A}_{adapt}^{xfw}$ has benign SRAM samples of any of the \emph{other} firmware running on the same MCU class, but none of the target firmware, and $\mathcal{A}_{adapt}^{1}$ has exactly one leaked benign SRAM sample of the target firmware (e.g., from a single side-channel exposure).

\begin{theorem}
Under $(A_2, A_4-A_7, A_{10})$, $\mathcal{A}$'s advantage is bounded by
\begin{equation}
\label{eq:obliv-adv}
\text{Adv}^{\text{IMP-obliv}}_{\textit{LiteAtt}}(\mathcal{A}) \leq FNR_{\text{rand}}(\eta) + \text{negl}(\lambda),
\end{equation}
\begin{equation}
\label{eq:adapt-adv}
\text{Adv}^{\text{IMP-adapt}}_{\textit{LiteAtt}}(\mathcal{A}) \leq FNR_{\text{adapt}}(\eta) + \text{negl}(\lambda),
\end{equation}
where $FNR_{\text{rand}}(\eta)$ is the empirically measured FNR of $M_{lite}$ under random byte perturbation of extent $\eta$ (Fig.~\ref{fig:sens-random}), and $FNR_{\text{adapt}}(\eta)$ is the FNR under adaptive perturbation (Fig. \ref{fig:sens-adapt-vs-obliv}).
\end{theorem}

\begin{proof}
$\mathcal{A}$ has four non-mutually-exclusive strategies, which we bound:

\textbf{\textit{(i) Evasion of $M_{lite}$:}} $\mathcal{A}$ executes $\mathcal{F}_m$ and relies on $M_{lite}$ producing $\gamma = 0$. The bound on this case depends on $\mathcal{A}$'s access to the model and SRAM.

For $\mathcal{A}_{obliv}$, the perturbation distribution is independent of $M_{lite}$. The DCT retains a fixed low-frequency band $[0, F)$, and for any oblivious perturbation $P$ over $\eta$-byte modifications, the expected energy in the retained band is upper-bounded by the energy induced by uniformly random byte replacement. Random-byte overwrite is therefore the structurally weakest perturbation, and any oblivious attack that perturbs $\geq \eta\%$ of bytes induces at least as much detectable DCT deviation in expectation. The empirical sensitivity analysis (Fig.~\ref{fig:sens-random}) thus yields:

\begin{equation}
FNR_{\text{rand}}(\eta) \leq
\begin{cases}
0.174 & \eta \geq 1\%, \\
0.042 & \eta \geq 2\%, \\
0.001 & \eta \geq 5\%, \\
0 & \eta \geq 10\%.
\end{cases}
\end{equation}

This argument does not hold for $\mathcal{A}_{adapt}$. By $A_4$, $\mathcal{A}_{adapt}$ has (a) no \textit{direct} query access to $M_{lite}, \mathcal{T}_{opt}$, and (b) at most one bit of feedback per session via the handshake outcome. With offline access to two authorized nodes, $\mathcal{A}$ may train a generator $\widehat{M}$ to construct evasive states. Empirical analysis of the three sub-tiers of $\mathcal{A}_{adapt}$ over 23 firmware, each with 30 malicious traces and $Q{=}5{,}000$ yields:
\begin{equation}
\label{eq:adapt-tiered}
FNR_{\text{adapt}}^{tier}(\eta) \leq
\begin{cases}
0 & tier \in \{none, xfw\},\ \forall \eta \in [1, 10]\%, \\
0.077 & tier = 1,\ \eta = 1\%, \\
0.125 & tier = 1,\ \eta = 2\%, \\
0.186 & tier = 1,\ \eta = 5\%, \\
0.287 & tier = 1,\ \eta = 10\%.
\end{cases}
\end{equation}

$\mathcal{A}_{adapt}^{none}$ and $\mathcal{A}_{adapt}^{xfw}$ achieve zero advantage. Specifically, $\widehat{M}$ in $\mathcal{A}_{adapt}^{xfw}$ fails since the SRAM seeds used to generate benign perturbations were unique to the source, highlighting the distribution differences between the SRAM dumps from different firmware. By gaining access to benign SRAM, only $\mathcal{A}_{adapt}^{1}$, accrues non-zero advantage, exceeding $FNR_{\text{rand}}$ by $\sim 3\!\times\!$ at $\eta\!=\!2\%$ and by orders of magnitude for $\eta\!\geq\!5\%$. Thus, we emphasise \emph{benign-SRAM extraction} as the vulnerability and recommend IoT vendors to block access to physical debug-ports to prevent exploitation. This residual advantage is an inherent limitation of all ML-based attestation primitives, including RAGE \cite{chilese2024one}, Swarm-Net \cite{kohli2024swarm}, SAFE-IoT \cite{kohli2024safe}, Aman et al. \cite{aman2022machine}, and Iqbal et al. \cite{iqbal2024ram}.

\textbf{\textit{(ii) Direct modification of $\gamma$:}} $\mathcal{A}$ attempts to overwrite $\gamma$ inside the TEE between inference and report packing. This is ruled out by TEE isolation under $A_4$, contributing at most $\text{negl}(\lambda)$.

\textbf{\textit{(iii) Transcript and report spoofing:}} $\mathcal{A}$ fabricates a handshake packet containing a dummy report $\mathcal{R}$ with $\gamma = 0$ without executing $\text{App}_{SA}$. Acceptance by $\mathcal{P}_j$ requires a valid transcript signature $\sigma_i \leftarrow \text{Sign}(epk_i\|N_1\|\mathcal{R}_i, SK_i)$, which is an EUF-CMA forgery experiment bounded by $\text{negl}(\lambda)$ under $A_2, A_5$.

\textbf{\textit{(iv) Replay-with-firmware-swap (TOCTOU):}} $\mathcal{A}$ could generate a valid $\mathcal{R}_t$ with $\gamma = 0$ under benign $\mathcal{F}$ and then swap to $\mathcal{F}_m$ during the session. However, under $A_{10}$, firmware replacement requires a device reset, which terminates the active session and destroys all ephemeral key material. $\mathcal{P}_j$ detects the disconnection. When $\mathcal{P}_i$ reconnects, a fresh handshake is initiated, generating a new attestation snapshot of the now-executing $\mathcal{F}_m$, which $M_{lite}$ detects with the bounds established in case (i). Consequently, no stale report can be injected into a live session, and this strategy contributes no additional advantage beyond (i). Combining (i)-(iv), strategies (ii)-(iv) collectively contribute at most $\text{negl}(\lambda)$ while strategy (i) dominates the residual advantage.
\end{proof}
 
\subsubsection{SRAM Privacy} (Game $G_{PRIV}$) $\mathcal{A}$ submits two SRAM contents $S_a, S_b$ with identical SA outcomes $\gamma(S_a) = \gamma(S_b)$. $\mathcal{C}$ selects $x \in \{a, b\}$, loads $S_x$ into the TEE, and runs $\text{App}_{SA}.\texttt{Generate}$. $\mathcal{A}$ observes the resulting network traffic and outputs $x'$.
 
\begin{theorem}
\label{thm:priv}
Under $(A_{4-6})$, $Adv^{PRIV}_{LiteAtt}(\mathcal{A}) = \left|Pr[x'=x] - \tfrac{1}{2}\right| \leq negl(\lambda)$.
\end{theorem}
\begin{proof}
The raw SRAM contents $S_x$ are processed entirely within the TEE under $A_4$ and never appear on any network message. The only $S_x$-dependent value communicated outside the TEE is the binary verdict $\gamma$, which by the game's construction is identical for both $S_a$ and $S_b$. The transmitted handshake elements contain $\gamma$, freshness metadata, and a transcript signature. None of these reveal information about the raw SRAM layout beyond what is explicitly encoded in the low-dimensional public fields. The privacy property captures the confidentiality of the SRAM contents sampled during attestation. Note that \textit{LiteAtt} does not hide the binary safe/unsafe verdict $\gamma$ from an authorized peer since the information revealed by it is redundant to the observed failure or success of the protocol.
\end{proof}

\subsection{Summary}
The protocol provides cryptographic security ($Adv\leq negl(\lambda)$) for mutual authentication, replay resistance, forward secrecy, transport-message confidentiality and integrity, and SRAM privacy. Attestation reports are bound to the session transcript via protocol-level signatures and explicit transcript context verification inside the TEE application, providing integrity and authenticity without requiring on-the-wire confidentiality of the binary outcome $\gamma$. Firmware impersonation remains a statistical guarantee whose tightness depends on the empirical robustness of $M_{lite}$ to adversarial inputs and the freshness parameters. We argue this is an inherent limit of any ML-based attestation primitive.

\section{Limitations} 
\label{sec:limitations} 
While TEEs provide significant security provisions \cite{jauernig2020trusted}, they remain prone to side-channel and physical attacks that are outside our scope. Further, side-channel attacks that extract information about the SRAM enable $M_{lite}$ spoofing. Mitigation such as adversarial training of section AEs, randomized DCT subspace selection, and constant-time, split-cache TEE implementations \cite{weissteiner2025teecorrelate} are promising directions to further harden the $\mathcal{A}_{adapt}^{1}$ case.

\section{Conclusion}
\label{sec:conclusion}

This paper presented \textit{LiteAtt}, a novel verifier-less, P2P-SA framework for IoT devices that leverage int8-quantized TinyAEs, SRAM runtime analysis, and Arm TrustZone TEE to provide on-device firmware integrity verification with stateless peers. Unlike prior remote attestation methods that rely on capable external verifiers or peers that match received evidence with reference states, \textit{LiteAtt} enables IoT devices to independently assess the integrity of their own firmware state and furnish secure, transcript-bound SA reports during routine connection handshakes with peer devices that do not need up-to-date reference states or ML models. \textit{LiteAtt} was validated on comprehensive SRAM datasets collected from real Arduino boards, maintaining an average 99.42\% accuracy, 99.70\% F1-score, 99.45\% TPR, and 95.14\% TNR while dropping full mutual attestation and key agreement overheads to a mere 26.3-294.9ms handshake latency, 2.65-9.35mJ energy consumption, and 4.91KB peak memory overhead across three Arm Cortex-M boards. The proposed decentralized approach enables IoT vendors to train models on twin hardware devices, thereby preserving localized user SRAM privacy and facilitating seamless OTA security updates. Furthermore, formal game-based evaluations verified that the \textit{LiteAtt} protocol guarantees mutual authentication, forward secrecy, transport confidentiality, data integrity, and strict resistance against real-time replay and TOCTOU impersonation attacks.

\bibliographystyle{IEEEtranN}
{\footnotesize\bibliography{references}}
\vskip -2\baselineskip plus -1fil
\begin{IEEEbiography}[{\includegraphics[width=1in,height=1.25in,clip,keepaspectratio]{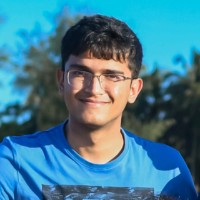}}]{Varun Kohli} is a Scientist at the Institute for Infocomm Research ($I^2R$), Agency of Science, Technology and Research (A*STAR), Singapore, and a Ph.D. student at the Department of Electrical and Computer Engineering at the National University of Singapore. He received his B.E. in Electrical and Electronics Engineering from the Birla Institute of Technology and Science, Pilani, India, in 2021. His research interests include Artificial Intelligence, IoT, and Cybersecurity.
\end{IEEEbiography}
\vskip -2\baselineskip plus -1fil
\begin{IEEEbiography}[{\includegraphics[width=1in,height=1.25in,clip,keepaspectratio]{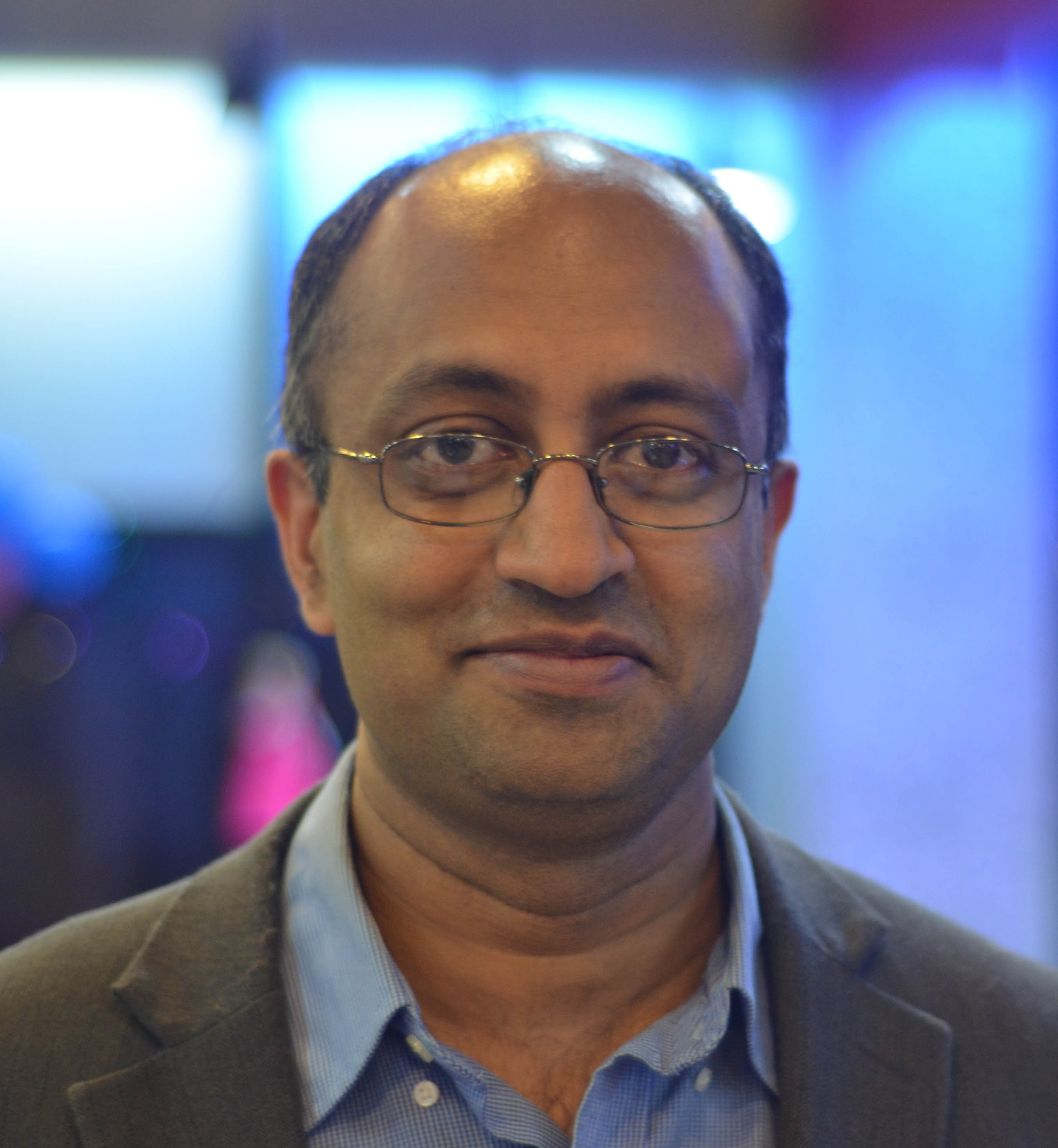}}]{Biplab Sikdar} received the B.Tech. degree in electronics and communication engineering from North Eastern Hill University, Shillong, India, in 1996, the M.Tech. degree in electrical engineering from the Indian Institute of Technology, Kanpur, India, in 1998, and the Ph.D. degree in electrical engineering from the Rensselaer Polytechnic Institute, Troy, NY, USA, in 2001. 
He is currently a Professor with the Department of Electrical and Computer Engineering, National University of Singapore, Singapore. His research interests include wireless network, and security for IoT and cyber-physical systems. 
\end{IEEEbiography}
\vfill
\end{document}